\documentclass[11pt]{article}
\usepackage{amsmath,amsthm,amssymb,enumerate,xcolor}
 \usepackage[dvips,pdfstartview=FitH,pdfpagemode=None,colorlinks=true,citecolor=blue,linkcolor=red]{hyperref} 
 \usepackage[ruled]{algorithm2e}
\newcommand{\remove}[1]{}
\newtheorem{thm}{Theorem}[section]  

\newtheorem{lem}[thm]{Lemma}
\newtheorem{define}[thm]{Definition}
\newtheorem{cor}[thm]{Corollary}
\newtheorem{obs}[thm]{Observation}

\newtheorem{construct}[thm]{Construction}

\newtheorem{fact}[thm]{Fact}
\newtheorem{prop}[thm]{Proposition}

\newtheorem{thm*}{Theorem}
\newtheorem*{cor*}{Corollary}
\newtheorem*{remark*}{Remark}

\def\F{{\mathbb{F}}}

\def\N{{\mathbb{N}}}
\def\R{{\mathbb{R}}}

\def\cB{{\mathcal B}}

\def\cH{\mathcal H}
\def\cF{\mathcal F}
\def\cG{\mathcal G}
\def\cC{\mathcal C}
\def\cA{\mathcal A}

\def\_{\,\,\,\,\,}

\def\then{\Rightarrow}

\def\D{{\partial}}

\def\mon{\textsf{mon}}

\def\poly{\textsf{poly}}
\def\var{\textsf{var}}

\def\VNP{\textsf{VNP}}

\newcommand{\eqdef}{{\stackrel{\rm def}{=}}}

\newcommand{\SPS}{\Sigma\Pi\Sigma} 
\newcommand{\tfSPS}[1]{\Sigma^{ [ #1 ] } \Pi \Sigma} 
\newcommand{\rsSPS}[1]{\Sigma\Pi\Sigma^{\{ #1 \}}} 
\newcommand{\tfrsSPS}[2]{\Sigma^{[ #1 ]} \Pi \Sigma^{\{ #2 \}}} 

\newcommand{\SPSP}{\Sigma\Pi\Sigma\Pi} 
\newcommand{\tfSPSP}[1]{\Sigma^{ [ #1 ] } \Pi \Sigma \Pi} 
\newcommand{\rsSPSP}[1]{\Sigma\Pi(\Sigma\Pi)^{\{ #1 \}}} 
\newcommand{\tfrsSPSP}[2]{\Sigma^{[ #1 ]} \Pi (\Sigma \Pi)^{\{ #2 \}}} 

\newcommand{\ol}[1]{\overline{#1}}
\newcommand{\valpha}{\ol{\alpha}}
\newcommand{\vbeta}{\ol{\beta}}
\newcommand{\Ext}{\mathcal{L}}

\newcommand{\eps}{\epsilon}

\newcommand{\ds}{\displaystyle}

\newcommand{\ignore}[1]{}

\allowdisplaybreaks

\begin{document}

\title{Subexponential Size Hitting Sets for Bounded Depth Multilinear Formulas}

\author{Rafael Oliveira\thanks{Department of Computer Science, Princeton University. Research supported by 
NSF grant CCF-1217416 and by the Sloan fellowship. Email: \texttt{rmo@cs.princeton.edu}.} \and
Amir Shpilka\thanks{Department of Computer Science, Tel Aviv University, Tel Aviv, Israel,
\texttt{shpilka@post.tau.ac.il}.  The research leading to these results has received funding
from the European Community's Seventh Framework Programme (FP7/2007-2013) under grant agreement number 257575, and 
from the Israel Science Foundation (grant number 339/10).} \and Ben Lee Volk\thanks{Department of Computer Science, Tel Aviv University, Tel Aviv, Israel,
\texttt{benleevolk@gmail.com}.  The research leading to these results has received funding
from the European Community's Seventh Framework Programme (FP7/2007-2013) under grant agreement number 257575.}}

\date{}
\maketitle

\begin{abstract}

In this paper we give subexponential size hitting sets for bounded depth multilinear arithmetic formulas. Using the known relation 
between black-box PIT and lower bounds we obtain lower bounds for these models. 

\sloppy For depth-$3$ multilinear formulas, of size $\exp(n^\delta)$, we give a hitting set of size $\exp\left(\tilde{O}\left(n^{2/3 + 2\delta/3} \right) \right)$. 
This implies a lower bound of $\exp(\tilde{\Omega}(n^{1/2}))$ for depth-$3$ multilinear formulas, for some explicit polynomial. 

\sloppy For depth-$4$ multilinear formulas, of size $\exp(n^\delta)$, we give a hitting set of size $\exp\left(\tilde{O}\left(n^{2/3 + 4\delta/3} \right) \right)$. 
This implies a lower bound of $\exp(\tilde{\Omega}(n^{1/4}))$ for depth-$4$ multilinear formulas, for some explicit polynomial. 

A regular formula consists of  alternating layers of $+,\times$ gates, where all gates at layer $i$ have the same fan-in. We give a 
hitting set of size (roughly) $\exp\left(n^{1- \delta} \right)$, for regular depth-$d$ multilinear formulas of size $\exp(n^\delta)$,
where $\delta = O(\frac{1}{\sqrt{5}^d})$. 
This result implies a lower bound of roughly $\exp(\tilde{\Omega}(n^{\frac{1}{\sqrt{5}^d}}))$ for such formulas. 

We note that better lower bounds are known for these models, but also that none of these bounds was achieved via construction of 
a hitting set. Moreover, no lower bound that implies such PIT results, even in the white-box model, is  currently known. 

Our results are combinatorial in nature and rely on reducing the underlying formula, first to a depth-$4$ formula, and then to a 
read-once algebraic branching program (from depth-$3$ formulas we go straight to read-once algebraic branching programs).

\end{abstract}
\thispagestyle{empty}

\newpage

\tableofcontents
\thispagestyle{empty}
\newpage
\pagenumbering{arabic}

\section{Introduction}

Arithmetic circuits are the standard model for computing polynomials. Roughly speaking, given a set of variables 
$X=\{x_1,...,x_n\}$, an arithmetic circuit uses additions and multiplications to compute a polynomial $f$ in the set of variables $X$. 
An arithmetic formula is an arithmetic circuit whose computation graph is a tree.
An arithmetic circuit (or formula) is multilinear if the polynomial computed at each of its gates is multilinear (as a formal polynomial), 
that is, in each of its monomials the power of every input variable is at most one (see Section~\ref{sec:def} for definition of the 
models studied in this paper)

Two outstanding open problems in complexity theory are to prove exponential lower bounds on the size of arithmetic circuits, i.e., 
to prove a lower bound on the number of operations required to compute some polynomial $f$, and to give efficient deterministic 
polynomial identity testing (PIT for short) algorithms for them. 
The PIT problem for arithmetic circuits asks the following question: given an arithmetic circuit $\Phi$  computing a polynomial $f$,  
determine, { \em efficiently and deterministically}, whether ``$f\equiv 0$''.  The black-box version of the PIT problem asks to construct 
a small {\em hitting set}, i.e., a set of evaluation points $\cH$, for which any such non-zero $f$ does not vanish on all the points in 
$\cH$.

It is  known that solving any one of the problems (proving lower bound or deterministic PIT), with appropriate parameters, 
for small depth (multilinear) formulas, is equivalent to solving it in the general (multilinear) case 
\cite{ValiantSkyumBR83,AgrawalVinay08,Koiran10,GuptaKKS13,Tavenas13}. It is also known that these two problems are tightly 
connected and that solving one would imply a solution to the other, both in the general case 
\cite{HeintzSchnorr80,KabanetsImpagliazzo03,Agrawal05} and in the bounded depth case\footnote{The result of \cite{DSY09} is 
more restricted than the results for circuits with no depth restrictions.} \cite{DSY09}. We note that in the multilinear case it is only 
known that hitting sets imply circuit lower bounds but not vice versa. 

In this work we study the PIT problem for several models of bounded depth multilinear formulas. Our main results are subexpoential 
size hitting sets for depth-$3$ and depth-$4$ multilinear formulas of subexponential size and for {\em regular} depth-$d$ multilinear 
formulas of subexponential size (with construction size deteriorating among the different models). 
Using the connection between explicit hitting sets and circuit lower bounds we get, as corollaries, subexponential lower bounds for 
these models. 


\subsection{Models for Computing Multilinear Polynomials}\label{sec:def}

An arithmetic circuit $\Phi$ over the field $\F$ and over the set of variables $X$ is a directed acyclic graph as follows. Every vertex in 
$\Phi$ of in-degree $0$ is labelled by either a variable in $X$ or a field element in $\F$. Every other vertex in $\Phi$ is labelled by 
either $\times$ or $+$. An arithmetic circuit is called a formula if it is a directed tree (whose edges are directed from the leaves to the 
root). The vertices of $\Phi$ are also called gates. Every gate of in-degree $0$ is called an input gate. Every gate of out-degree $0$ 
is called an output gate. Every gate labelled by $\times$ is called a product gate. Every gate labelled by $+$ is called a sum gate. 
An arithmetic circuit computes a polynomial in a natural way. An input gate labelled by $y \in \F \cup X$ computes the polynomial $y$. 
A product gate computes the product of the polynomials computed by its children. A sum gate computes the sum of the polynomials 
computed by its children.

A polynomial $f\in \F[X]$ is called multilinear if the degree of each variable in $f$ is at most one. An arithmetic circuit (formula) $\Phi$ 
is called multilinear if every gate in $\Phi$ computes a multilinear polynomial.

In this work we are interested in small depth multilinear formulas. 
A depth-$3$ $\Sigma\Pi\Sigma$ formula is a formula composed of three layers of alternating sum and product gates. Thus, every 
polynomial computed by a $\Sigma\Pi\Sigma$ formula of size $s$ has the following form $$f = \sum_{i=1}^{s}\prod_{j=1}^{d_i}\ell_{i,j},$$
where the $\ell_{i,j}$ are linear functions. In a $\Sigma\Pi\Sigma$ multilinear formula, it holds that in every product gate, 
$\prod_{j=1}^{d_i}\ell_{i,j}$, the linear functions $\ell_{i,1},\ldots,\ell_{i,d_i}$ are supported on disjoint sets of variables.

Similarly, a depth-$4$ $\Sigma\Pi\Sigma\Pi$ formula is a formula composed of four layers of alternating sum and product gates. 
Thus, every polynomial computed by a $\Sigma\Pi\Sigma\Pi$ formula of size $s$ has the following form 
$$f = \sum_{i=1}^{s}\prod_{j=1}^{d_i}Q_{i,j},$$
where the $Q_{i,j}$ are computed at the bottom $\Sigma\Pi$ layers and are $s$-sparse polynomials, i.e., polynomials that have at 
most $s$ monomials. As in the depth-$3$ case, we have that at every product gate the polynomials $Q_{i,1},\ldots,Q_{i,d_i}$  are 
supported on disjoint sets of variables.
%
 
Another important definition for us is that of a regular depth-$d$ formula.
A regular depth-$d$ formula is specified by a list of $d$ integers $(a_1,p_1,a_2,p_2, \ldots)$. It has $d$ layers of alternating sum and 
product gates. The fan-in of every sum gate at the $(2i-1)$'th layer is $a_i$ and, similarly, the fan-in of every product gate at the 
$(2i)$'th layer is $p_i$. For example, a depth-$4$ formula that is specified by the list $(a_1,p_1,a_2,p_2)$ has the following form:
$$f = \sum_{i=1}^{a_1}\prod_{j=1}^{p_1}Q_{i,j},$$
where each $Q_{i,j}$ is a polynomial of degree $p_2$ that has (at most) $a_2$ monomials. 
As before, a  regular depth-$d$ multilinear formula is a  regular depth-$d$ formula in which every gate computes a multilinear polynomial.

Regular formulas where first defined by Kayal, Saha and Saptharishi \cite{KayalSS14}, who proved quasi-polynomial lower bounds 
for logarithmic-depth regular formulas. It is interesting to note that in the reductions from general (multilinear) circuits/formulas to 
depth-$d$ (multilinear) formulas, one gets a regular depth-$d$ (multilinear) formula 
\cite{ValiantSkyumBR83,AgrawalVinay08,Koiran10,Tavenas13}.

Finally, we also need to consider the model of Read-Once Algebraic Branching Programs (ROABPs) as our construction is based 
on a reduction to this case. Algebraic branching programs were first defined in the work of Nisan~\cite{Nisan91} who proved 
exponential lower bounds on the size of non-commutative ABPs computing the determinant or permanent polynomials. Roughly, an 
algebraic branching program (ABP) consists of a layered graph with edges going from the $i$'th layer to the $(i+1)$'th layer. The first 
layer consists of a single source node and the last layer contains a single sink. The edges of the graph are labeled with polynomials 
(in our case we only consider linear functions as labels). The weight of a path is the product of the weights of the edges in the path. 
The polynomial computed by the ABP is the sum of the weights of all the paths from the source to the sink. An ABP is called a 
read-once ABP (ROABP) if the only variable appearing on edges that connect the $i$'th and the $(i+1)$'th layer is $x_i$. It is clear 
that a ROABP whose edges are labeled with linear functions computes a multilinear polynomial.


\subsection{Polynomial Identity Testing}\label{sec:PIT}

In the PIT problem we are given an arithmetic circuit or formula $\Phi$, computing some polynomial $f$, and we have to determine 
whether ``$f\equiv 0$''.  That is, we are asking if $f$ is the zero polynomial in $\F[x_1,\ldots,x_n]$. By the 
Schwartz-Zippel-DeMillo-Lipton lemma \cite{Zippel79,Schwartz80,DemilloL78}, if $0\ne f \in \F[x_1,\ldots,x_n]$ is a polynomial of 
degree $\le d$, and $\alpha_1,\ldots,\alpha_n \in A\subseteq\F$ are chosen uniformly at random, then $f(\alpha_1,\ldots,\alpha_n) =0$ 
with probability at most\footnote{Note that this is meaningful only if $d < |A| \leq |\F|$, which in particular implies that $f$ is not the zero 
function.} $d/|A|$. Thus, given $\Phi$, we can perform these evaluations efficiently, giving an efficient randomized procedure for 
answering ``$f\equiv 0$?''.  It is an important open problem to find a derandomization of this algorithm, that is, to find a 
{\em deterministic} procedure for PIT that runs in polynomial time (in the size of $\Phi$).

One interesting property of the above randomized algorithm of Schwartz-Zippel is that the algorithm does not need to ``see'' the circuit 
$\Phi$. Namely, the algorithm only uses the circuit to compute the evaluations $f(\alpha_1,\ldots,\alpha_n)$.  Such an algorithm is 
called a {\em black-box} algorithm. In contrast, an algorithm that can access the internal structure of the circuit $\Phi$ is called a 
{\em white-box} algorithm.  Clearly, the designer of the algorithm has more resources in the white-box model and so one can expect 
that solving PIT in this model should be a simpler task than in the black-box model.

The problem of derandomizing PIT has received a lot of attention in the past few years. In particular, many works examine a specific 
class of circuits $\cC$, and design PIT algorithms only for circuits in that class. One reason for this attention is the strong connection 
between deterministic PIT algorithms for a class $\cC$ and lower bounds for $\cC$. This connection was first observed by Heintz and 
Schnorr~\cite{HeintzSchnorr80} (and later also by Agrawal~\cite{Agrawal05}) for the black-box model and by Kabanets and 
Impagliazzo~\cite{KabanetsImpagliazzo04} for the white-box model (see also Dvir, Shpilka and Yehudayoff~\cite{DSY09} for a similar 
result for bounded depth circuits). Another motivation for studying the problem is its relation to algorithmic questions. Indeed, the 
famous deterministic primality testing algorithm of Agrawal, Kayal and Saxena~\cite{AKS04} is based on derandomizing a specific 
polynomial identity. Finally, the PIT problem is, in some sense, the most general problem that we know today for which we have 
randomized coRP algorithms but no polynomial time algorithms, thus studying it is a natural step towards a better understanding of 
the relation between RP and P. For more on the PIT problem we refer to the survey by Shpilka and Yehudayoff~\cite{SY10}.

Among the most studied circuit classes we  find Read-Once Algebraic Branching Programs 
\cite{ForbesShpilka13,ForbesSS14,AgrawalGKS14}, set-multilinear formulas 
\cite{RazShpilka05,ForbesShpilka12,AgrawalSS13}, depth-$3$ formulas 
\cite{DvirShpilka06,KayalSaxena07,KarninShpilka08,KayalSaraf09,SaxenaSeshadhri11}, multilinear depth-$4$ formulas (and some 
generalizations of them) \cite{KarninMSV13,SarafVolkovich11,AgrawalSSS12} and bounded-read multilinear formulas 
\cite{ShpilkaVolkovich08,ShpilkaVolkovich09,AvMV11,AgrawalSSS12}. We note that none of these results follow from a reduction a 
la Kabanets-Impagliazzo \cite{KabanetsImpagliazzo04} (or the reduction of \cite{DSY09} for bounded depth circuits) from PIT to 
lower bounds. Indeed, this reduction does not work for the restricted classes mentioned here. In particular, for the multilinear model 
no reduction from PIT to lower bounds is known. That is, even given lower bounds for multilinear circuits/formulas (e.g., the 
exponential lower bound of Raz and Yehudayoff \cite{RazYehudayoff09} for constant depth multilinear formulas) we do not know how 
to construct a PIT algorithm for a related model. 

The works on depth-$3$ and multilinear depth-$4$ formulas gave polynomial time algorithms only when the fan-in of the top gate 
(the output gate) is constant, and became exponential time when the top fan-in was $\Omega(n)$, both in the white-box and black-box 
models \cite{KayalSaxena07,SaxenaSeshadhri11,SarafVolkovich11}. Raz and Shpilka \cite{RazShpilka05} gave a polynomial time PIT 
for set-multilinear depth-$3$ circuits and Forbes and Shpilka  \cite{ForbesShpilka13} and Agrawal, Saha and Saxena \cite{AgrawalSS13} 
gave a quasi-polynomial size hitting set for the model. Recall that in a depth-$3$ set-multilinear formula, the variables are partitioned 
to sets, and each linear function at the bottom layer only involves variables from a single set. Recently, Agrawal et 
al.\ \cite{AgrawalGKS14} gave a subexponential white-box algorithm for a depth-$3$ formula that computes the sum of $c$ 
set-multilinear formulas, each of size $s$, with respect to different partitions of the variables. The running time of their algorithm is 
$n^{O(2^{c} n^{1-\frac{2}{2^c} }\log s)}$. In particular, for $c=O(\log\log(n))$ the running time is $\exp(n)$.

Thus, prior to this work there were no subexponential PIT algorithms, even for depth-$3$ multilinear formulas with top fan-in $n$.

\subsection{Our Results}\label{sec:results}

\begin{remark*}
Throughout this paper, we assume that for formulas of size $2^{n^{\delta}}$, the underlying field $\F$ is of size at least $|\F| \ge 2^{n^{2\delta} \poly\log(n)}$, and that if this is not the case then we are allowed to query the formula on inputs from an extension field of the appropriate size. In particular, all our results hold over fields of characteristic zero or over fields of size $\exp(n)$.
\end{remark*}

We give subexponential size hitting sets for depth-$3$, depth-$4$ and regular depth-$d$ multilinear formulas, of subexponential size. In particular we obtain the following results.

\begin{thm}
\label{thm:intro:hitting-set-depth-$3$}
There exists a hitting set $\cH$ of size $|\cH| = 2^{\tilde{O}(n^{\frac{2}{3}+\frac{2}{3}\delta})}$ for the class of $\SPS$ multilinear 
formulas of size $2^{n^\delta}$.
\end{thm}

This  gives a significant improvement to the recent result, mentioned above, of Agrawal et al.\ \cite{AgrawalGKS14} who studied 
sum of set-multilinear formulas. 
From the connection between hitting sets and circuit lower bounds \cite{HeintzSchnorr80,Agrawal05} we obtain the following corollary.

\begin{cor}\label{cor:intro:lowerbound-depth-$3$}
There is an explicit multilinear polynomial $f\in \F[x_1,\ldots,x_n]$, whose coefficients can be computed in exponential time, such that 
any depth-$3$ multilinear formula for $f$ has size $\exp(\tilde{\Omega}(\sqrt{n}))$.
\end{cor}

This lower bound is weaker than the exponential lower bound of Nisan and Wigderson for this model \cite{NisanWigderson96}. Yet, 
it is interesting to note that we can get such a strong lower bound from a PIT algorithm. 
Next, we present our result for depth-$4$ multilinear formulas. 

\begin{thm}
\label{thm:intro:hitting-set-depth-$4$}
There exists a hitting set $\cH$ of size $|\cH| = 2^{\tilde{O}(n^{2/3+4\delta/3})}$ for the class of $\SPSP$ multilinear formulas of 
size $2^{n^\delta}$.
\end{thm}

Again, from the connection between hitting sets and circuit lower bounds we obtain the following corollary.

\begin{cor}\label{cor:intro:lowerbound-depth-$4$}
There is an explicit multilinear polynomial $f\in \F[x_1,\ldots,x_n]$, whose coefficients can be computed in exponential time, such that 
any depth-$4$ multilinear formula for $f$ has size $\exp(\tilde{\Omega}(n^{1/4}))$.
\end{cor}

The best known lower bound for depth-$4$ multilinear formulas is $\exp(n^{1/2})$ due to Raz and Yehudayoff \cite{RazYehudayoff09}, 
thus, as in the previous case, the term in the exponent of our lower bound is the square root of the corresponding term in the best 
known lower bound. For regular depth-$d$ multilinear formulas we obtain the following result.

\begin{thm}
\label{thm:intro:hitting-set-regular}
There exists a hitting set $\cH$ of size $|\cH| = 2^{\tilde{O}(n^{1- \delta/3})}$ for the class of regular depth-$d$ multilinear 
formulas of size $2^{n^\delta}$, where $\delta \leq  \frac{1}{5^{\lfloor d/2 \rfloor+1}}  = O\left(\frac{1}{\sqrt{5}^d}\right)$.
\end{thm}

As before we obtain a lower bound for such formulas.

\begin{cor}\label{cor:intro:lowerbound-regular}
There is an explicit multilinear polynomial $f\in \F[x_1,\ldots,x_n]$, whose coefficients can be computed in exponential time, 
such that any regular depth-$d$ multilinear formula for $f$ has size $\exp(\tilde{\Omega}(n^{ \frac{1}{5^{\lfloor d/2 \rfloor+1}} }))$.
\end{cor}

We note that Raz and Yehudayoff gave an $\exp(n^{\Omega(\frac{1}{d})})$ lower bound  for depth-$d$ multilinear formulas, which 
is much stronger than what Corollary~\ref{cor:intro:lowerbound-regular} gives. Yet, our result also gives a PIT algorithm, which does 
not follow from the results of \cite{RazYehudayoff09}. As we later explain, we lose a square root in the term at 
the exponent for every increase of the depth and this is the reason that we get only $\exp(n^{1/\exp(d)})$ instead of $\exp(n^{1/d})$.

In addition to lower bounds, our work also implies deterministic factorization of multilinear polynomials. 
In~\cite{ShpilkaVolkovich10}, Shpilka and Volkovich proved that if one can perform PIT deterministically for certain classes of multilinear polynomials then a deterministic factoring algorithm for those classes follows. Specifically, for a class of polynomials $\cC$ they defined the  class $\cC_V$, consisting of all polynomials that can be computed by circuits of the form $C = C_1 + C_2 \times C_3$, where the circuits $C_i$ belong to the class $\cC$ and the circuits $C_2$ and $C_3$ are defined over disjoint sets of variables. They proved that if the class $\cC_V$ has a deterministic PIT that runs in time $T(n,s)$ for circuits on $n$  variables of size $s$ then there is a deterministic factoring algorithm for the class $\cC$ that runs in time $O(n^3\cdot T(s))$ (Theorem 1.1 in \cite{ShpilkaVolkovich10}). 

In our case, since the product of two variable disjoint multilinear $\SPS$ ($\SPSP$) formulas of size $2^{n^\delta}$ is a 
multilinear $\SPS$ ($\SPSP$) formula of size $2^{2n^\delta}$, which is still inside of the class $\SPS$ ($\SPSP$), the result of \cite{ShpilkaVolkovich10}, when combined with our PIT results, implies that we can deterministically factor such formulas.  
Therefore, we obtain the following corollary:

\begin{cor}[Deterministic Factorization]\label{cor:intro:factoring}
	Given a multilinear polynomial $f \in \F[x_1, \ldots, x_n]$ that can be computed by a $\SPS$ ($\SPSP$) formula of size $2^{n^\delta}$, there exists an efficient deterministic algorithm that outputs the factors of $f$. The algorithm outputs $\SPS$ ($\SPSP$) formulas for the factors if the formula for $f$ is given to it explicitly, and black-boxes if it only has black-box access to $f$. The running time of this algorithm is $2^{\tilde{O}(n^{\frac{2}{3}+ \frac{2}{3} \delta})}$ when $f$ is computed by a $\SPS$ formula and $2^{\tilde{O}(n^{\frac{2}{3}+ \frac{4}{3} \delta})}$ when it is computed by a $\SPSP$ formula. 
\end{cor}


\subsection{Proof Overview}\label{sec:technique}

We first discuss our proof technique for the case of depth-$3$ multilinear formulas. 
Our (idealized) aim is to reduce such a formula $\Phi$ to a depth-$3$ multilinear formula in which each linear function is of the form 
$\alpha x + \beta$. That is, each linear function contains at most one variable. If we manage to do that then we can use the 
quasi-polynomial sized hitting set of \cite{ForbesShpilka12,AgrawalGKS14} for this model.

Of course, the problem with the above argument is that in general, depth-$3$ formulas have more than one variable per linear function. 
To overcome this difficulty, we will partition the variables to several sets $T_1,\ldots,T_m$ and hope that each linear function in the 
formula contains at most one variable from each $T_i$. If we can do that then we would use the hitting set for each set of variables 
$T_i$ and combine those sets together to get our hitting set. That is, the combined hitting set is composed of concatenation of all 
vectors of the form $v_1 \circ v_2 \circ \ldots \circ v_m$ where $v_i$ comes from the hitting set restricted to the variables in $T_i$ 
(the concatenation is performed in a way that respects the partition of course). Thus, if we can carry out this procedure then we will 
get a hitting set of size roughly $n^{m\log n}$. This step indeed yields a hitting set, since when we restrict our attention to each $T_i$ 
and think of the other variables as constants in some huge extension field, then we do get a small ROABP (in the variables of $T_i$) 
and hence plugging in the hitting set of  \cite{ForbesShpilka12,AgrawalGKS14} gives a non-zero polynomial. Thus, we can first do this 
for $T_1$ and obtain some good assignment $v_1$ that makes the polynomial non-zero after substituting $v_1$ to $T_1$. Then we 
can find $v_2$ etc.

There are two problems with the above argument. One problem is how to find such a good partition. The second is that this idea 
simply cannot work as is. For example, if we have the linear function $x_1+\cdots+x_n$, then it will have a large intersection with 
each $T_i$. 

We first deal with the second question. to overcome the difficulty posed by the example, we would like to somehow ``get rid'' of all 
linear functions of large support and then carry out the idea above. To remove linear functions with large support from the formula we 
use another trick. Consider a variable $x_k$ that appears in a linear function $\ell_0$ that has a large support. Assume that 
$\frac{\partial f}{\partial x_k} \not\equiv 0$ as otherwise we can ignore $x_k$. Now, because of multilinearity, we can transform 
our original formula $\Phi$ to a formula computing $\frac{\partial f}{\partial x_k}$. This is done by replacing each linear function 
$\ell(X) = \sum_{i=1}^{n}\alpha_i x_i + \alpha_0$ with the constant $\alpha_k$.
In particular, the function $\ell_0$ that used to have a high support does not appear in the new formula. Furthermore, this process 
does not increase the support size of any other linear function. A possible issue is that if we have to repeat this process for every 
function of large support then it seems that we need to take a fresh derivative for every such linear function. The point is that because 
we only care about linear functions that have a large support to begin with, we can find a variable that simultaneously appears in 
many of those functions and thus one derivative will eliminate many of the ``bad'' linear functions. 
Working out the parameters, we see that we need to take roughly $n^\epsilon \cdot \log|\Phi|$ many derivatives to reduce to the 
case where all linear functions have support size at most $n^{1-\epsilon}$.

Now we go back to our first problem. We can assume that we have a depth-$3$ formula in which each linear function has support 
size at most $n^{1-\epsilon}$ and we wish to find a partition of the variables  to  sets $T_1,\ldots,T_m$ so that each $T_i$ contains 
at most one variable from each linear function. This cannot be achieved as a simple probabilistic argument shows, so we relax our 
requirement and only demand that in each multiplication gate (of the formula) only a few linear functions  have a large intersection. 
If at most $k$ linear functions in each gate have a large intersection, we can expand each multiplication 
gate to at most $n^k$ new gates (by simply expanding all linear functions that have large intersection) and then apply our argument. 
As we will be able to handle subexponential size formulas, this blow up is tolerable for us.

Note that if we were to pick the partition at random, when $m= n^{1-\epsilon+\gamma}$, for some small $\gamma$, then we will get 
that with very high probability at most $n^\delta$ linear functions will have interaction at most $n^\delta$ with each $T_i$, where 
$\delta$ is such that $|\Phi| < \exp(n^\delta)$. To get a deterministic version of this partitioning, we simply use an $n^\delta$-wise 
independent family of hash functions $\{h:[n]\to [m]\}$. Each hash function $h$ induces a partition of the variables to 
$T_i = \{x_k \mid h(k)=i\}$. Because of the high independence, we are guaranteed that there is at least one hash function that  
induces a good partition. 

Now we have all the ingredients in place. To get our hitting set we basically do the following (we describe the construction as a 
process, but it should be clear that every step can be performed using some evaluation vectors).
\begin{enumerate}

\item \label{item:derivative} Pick $n^\epsilon \cdot \log|\Phi|$ many variables and compute a black-box for the polynomial that is 
obtained by taking the derivative of $f$ with respect to  those variables. The cost of this step is roughly 
${n \choose n^\epsilon \cdot \log|\Phi|} \cdot 2^{n^\epsilon \cdot \log|\Phi|}$, where the first term is for picking the variables and 
the second is what we have to pay to get access to the derived polynomial.

\item \label{item:partition} Partition the remaining variables to (roughly) $n^{1-\epsilon/2}$ many sets using a (roughly) $\log|\Phi|$-wise 
independent family of hash functions. The cost of this step is roughly $n^{\log|\Phi|}$ as this is the size of the hash function family.

\item Plug in a fresh copy of the hitting set of \cite{ForbesShpilka12,AgrawalGKS14} to each of the sets of variables $T_i$. The cost 
is roughly $n^{\log n \cdot n^{1-\epsilon/2}}$.

\end{enumerate}

Combining everything we get a hitting set of size roughly 
$$\left({n \choose n^\epsilon \cdot \log|\Phi|} \cdot 2^{n^\epsilon \cdot \log|\Phi|}\right)\cdot \left(n^{\log|\Phi|}\right)\cdot \left(n^{\log n \cdot n^{1-\epsilon/2}}\right) \approx 2^{\tilde{O}\left( n^{1-\epsilon/2} + n^\epsilon \log|\Phi| \right)}.$$
Optimizing the parameters we get our hitting set.

We would like to use the same approach also for the case of depth-$4$ formulas. Here the problem is that in the two bottom layers 
the formula compute a polynomial and not a linear function. In particular, when taking a derivative we are no longer removing 
functions that have a large support. Nevertheless, we can still use a similar idea. We show  there is a variable $x_i$ that by either 
setting $f|_{x_i=0}$ or considering $\frac{\partial f}{\partial x_i}$, we are guaranteed that the total sparsity of all polynomials that 
have large supports goes down by some non-negligible factor. Thus, repeating this process (of either setting a variable to $0$ or 
taking a derivative) $n^\epsilon \cdot \log|\Phi|$ many times we reach a depth-$4$ formula where all polynomials computed at the 
bottom addition gate have small supports. Next, we partition the variables to buckets and consider a single bucket $T_i$. Now, 
another issue is that even if the intersection of a low-support polynomial with some $T_i$ is rather small, the sparsity of the 
resulting polynomial (which is considered as a polynomial in the variable in the intersection) can still be exponential in the size of 
the intersection. This is why we lose a bit in the upper bound compared to the depth-$3$ case. Combining all steps again we get 
the result for depth-$4$ formulas.

The proof for regular formulas works by first reducing to the depth-$4$ case and then applying our hitting set. The reduction is 
obtained in a similar spirit to the reduction of Kayal et al.\ \cite{KayalSS14}. We break the formula at an appropriate layer and then 
express the top layers as a $\Sigma\Pi$ circuit and the bottom layers as products of polynomials of not too high degrees. We then 
use the trivial observation that if the degree of a polynomial is at most $n^{1-\epsilon}$ then its sparsity is at most $n^{n^{1-\epsilon}}$ 
and proceed as before. Due to the different requirements of the reduction and of the hashing part, we roughly lose a constant 
factor in the exponent of $n$, in the size of the hitting set, whenever the depth grows, resulting in a hitting set of size roughly 
$\exp(n^{1-1/\exp(d)})$.

To obtain the lower bounds we simply use the idea of \cite{HeintzSchnorr80,Agrawal05}. That is, given a hitting set we find a 
non-zero multilinear polynomial that vanishes on all points of the hitting set by solving a homogeneous system of linear equations. 


\subsection{Related Work}

\paragraph{The work of Agrawal, Gurjar, Korwar and Saxena \cite{AgrawalGKS14}:}

The closest work to ours is the one by Agrawal et al.\ \cite{AgrawalGKS14}. In addition to other results, they gave a white-box PIT 
algorithm that runs in time $n^{O(2^{c} n^{1-\frac{2}{2^c} }\log s)}$ for depth-$3$ formulas that can be represented as a sum of 
$c$ set-multilinear formulas, each of size $s$ (potentially with respect to different partitions of the variables).

Theorem~\ref{thm:intro:hitting-set-depth-$3$} improves upon this results in several ways. First, the theorem gives a hitting set, 
i.e., a black-box PIT. Secondly, for  $c=O(\log\log n)$ the running time of the algorithm of \cite{AgrawalGKS14} is $\exp(n)$, whereas 
our construction can handle a sum of $\exp(n^\beta)$ set-multilinear formulas and still maintain a subexponential complexity. 

Nonetheless, there are some similarities behind the basic approach of the this work and the work of Agrawal et al. Recall that a set-multilinear depth-$3$ formula is based on a partition of the variables, where each linear function in the formula contains variables from a single partition. 
Agrawal et al.\ start with a sum of $c$ set-multilinear circuits, each with respect to a different partitioning of the variables, and their first goal is to reduce the formula to a set-multilinear formula, i.e., to have only one partition of the variables. For this they define a distance between different partitions and show, using an involved combinatorial argument, that one can find some partition $T_1,\ldots,T_m$ of the variables so that when restricting our attention to $T_i$, all the $c$ set-multilinear formulas will be somewhat ``close to each other''. If the distance is $\Delta$ (according to their definition) then they prove that they can express the sum as a set-multilinear circuit of size roughly $s\cdot n^\Delta$, where $s$ is the total size of the depth-$3$ formula. Unlike our work, they find the partition in a white-box manner by gradually refining the given $c$ partitions of the set-multilinear circuits composing the formula. The final verification step is done, in a similar manner to ours, by substituting the hitting set of \cite{AgrawalSS13} (or that of \cite{ForbesShpilka12}) to each of the sets $T_i$. The step of finding the partition $T_1,\ldots,T_m$ is technically involved and is the only step where white-box access is required. 

\paragraph{Lower bounds for multilinear circuits and formulas:}

Lower bounds for the multilinear model were first proved by Nisan and Wigderson \cite{NisanWigderson96}, who gave exponential lower bounds for depth-$3$ formulas. Raz first proved quasi-polynomial lower bounds for multilinear formulas computing the Determinant and Permanent polynomials \cite{Raz09a} and later gave a separation
 between multilnear $\text{NC}_1$ and multilinear $\text{NC}_2$ \cite{Raz06}. 
Raz and Yehudayoff proved a lower bound of $\exp(n^{\Omega(\frac{1}{d})})$ for depth-$d$ multilinear formulas. 
As in the general case, the depth reduction techniques of \cite{ValiantSkyumBR83,AgrawalVinay08,Koiran10,Tavenas13} also work for multilinear formulas. Thus, proving a lower bound of the form $\exp(n^{\frac{1}{2}+\epsilon})$ for $\SPSP$ multilinear formulas, would imply a super-polynomial lower bound for multilinear circuits. Currently, the best lower bound for syntactic multilinear circuits is $n^{4/3}$ by Raz, Shpilka and Yehudayoff \cite{RSY08}.

Kayal, Saha and Saptharishi \cite{KayalSS14} proved a quasi-polynomial lower bounds for regular formulas that have the additional condition that the syntactic degree of the formula is at most twice the degree of the output polynomial.

\subsection{Organization}
In Section~\ref{sec:prelim} we provide basic definitions and notations, and also prove some general lemmas which will be helpful in the next sections. In Section~\ref{sec:red}, we explain how to reduce general depth-$3$ and depth-$4$ formulas to formulas such that every polynomial at the bottom has small support. Then, in Section~\ref{sec:hit-bottom}, we construct a hitting set for those types of formulas. In Section~\ref{sec:hit}, we explain how to combine the ideas of the previous two sections and construct our hitting set for depth-$3$ and depth-$4$ multilinear formulas.

We then move on in Section~\ref{sec:regular} to depth $d$ regular formulas, and show how to reduce them to depth-$4$ formulas and obtain a hitting set for this class. In the short Section~\ref{sec:lower-bounds} we spell out briefly how, using known observations about the relation between PIT and lower bounds, we obtain our lower bounds for multilinear formulas. Finally, in Section~\ref{sec:open} we discuss some open problems and future directions for research.


\section{Preliminaries}
\label{sec:prelim}

In this section, we establish notation, some definitions and useful lemmas that will be used throughout the paper.

\subsection{Notations and Basic Definitions}\label{sec:notation}

For any positive integer $n$, we denote by $[n]$ the set of integers from $1$ to $n$, and by $\binom{[n]}{\le r}$ the family of subsets 
$A \subseteq [n]$ such that $|A| \le r$. We often associate a subset $A \subseteq [n]$ with a subset of variables 
$\var(A) \subseteq \{ x_1,\ldots,x_n \}$ in a natural way (i.e., $\var(A)=\{x_i \mid i \in A \}$). In those cases we make no 
distinction between the two and use $A$ to refer to $\var(A)$.
Additionally, if $A$ and $B$ are disjoint subsets of $[n]$, we denote their disjoint union by $A \sqcup B$. For a vector $v\in\F^n$ we denote with $v|_A$ the restriction of $v$ to the coordinates $A$.

In order to improve the readability of the text, we omit floor and ceiling notations.

Let $f(x_1, \ldots, x_n) \in \F[x_1, \ldots, x_n]$ be a polynomial. We will denote by $\D_x f$ the formal derivative
of $f$ with respect to the variable $x$, and by $f|_{x=0}$ the polynomial obtained from $f$ by setting $x = 0$. Moreover, 
if $A \subseteq [n]$, we will denote by $\D_A f$ the polynomial obtained when taking the formal derivative of $f$ with respect
to all variables in $A$. In a similar fashion, we denote by $f|_{A=0}$ the polynomial obtained when we set all the variables in $A$
to zero, and more generally, if $|A|=r$ and $\valpha = (\alpha_1,\ldots,\alpha_r) \in \F^{r}$, $f|_{A=\valpha}$ will denote the restriction of $f$ obtained when setting the $i$'th variable in $A$ to $\alpha_i$, for $1 \le i \le r$. 

In addition to the conventions above, the following definitions will be very useful in the next sections.

\begin{define}[Variable Set and Non-trivial Variable Set]
	Let $f(x_1, \ldots, x_n) \in \F[x_1, \ldots, x_n]$ be a polynomial. Define the variable set ($\var$) and the non-trivial
	variable set ($\var^*$) as follows:
	\begin{align*}  
		\var(f)     &= \{ x \in \{x_1, \ldots, x_n\} \mid \D_x f \neq 0  \}  \\
		\var^*(f) &= \{ x \in \{x_1, \ldots, x_n\} \mid \D_x f \neq 0 \text{ and } f|_{x=0} \neq 0  \}.
	\end{align*}
	That is, the variable set of a polynomial $f$ is the set of variables $x \in \{x_1, \ldots, x_n\}$ that appear in the
	representation of $f$ as a sum of monomials, whereas the non-trivial variable set is the set of variables of $f$ 
	that do not divide it. 
\end{define}

We shall say that $f$ has a small support if $\var(f)$ (or $\var^*(f)$) is not too large. 

\begin{define}[Monomial Support and Sparsity]
	Let $f(x_1, \ldots, x_n) \in \F[x_1, \ldots, x_n]$ be a polynomial. We define the \emph{monomial support} of $f$, written 
	$\mon(f)$, as the set of monomials that have a non-zero coefficient in $f$. In addition, we define the sparsity of $f$, written 
	$\| f \|$, as the size of the set $\mon(f)$, that is,
	$$  \| f \| = |\mon(f)|. $$
	In other words, the sparsity of $f$ is the number of monomials that appear with a non-zero coefficient in $f$.
\end{define}

In the constructions of our hitting sets we will need to combine assignments to different subsets of variables. The following notation will be useful. For a partition of $[n]$, $T_1\sqcup T_2\sqcup \cdots \sqcup T_m =[n]$, and sets $\cH_i \subseteq \F^{|T_i|}$, we denote with $\cH_1^{T_1}\times \cdots \times \cH_m^{T_m}$ the set of all vectors of length $n$ whose restriction to $T_i$ is an element of $\cH_i$:
$$\cH_1^{T_1}\times \cdots \times \cH_m^{T_m} = \{ v \in\F^n \mid \forall i\in [m], \; v|_{T_i} \in \cH_i  \}.$$ 

\subsection{Depth-$3$ and Depth-$4$ Formulas}

We define some special classes of depth-$3$ and depth-$4$ formulas that will be used throughout this paper.

\begin{define}[Restricted Top Fan-in]
	Let $\Phi$ be a multilinear depth-$4$ formula. We say that $\Phi$ is a multilinear $\tfSPSP{M}$ 
	formula if it is of the form $\sum_{i=1}^m\prod_{j=1}^{t_i} f_{i,j}$, where $m \le M$.
	If, in addition to the conditions above, we have that each $f_{i,j}$ is a linear function, that is, $\Phi$ is actually a 
	depth-$3$ formula, we will say that $\Phi$ is a multilinear $\tfSPS{M}$ formula.
\end{define}

Our next definition considers the case where polynomials computed at the bottom layers do not contain too many variables, that is, they have small support. 

\begin{define}[Restricted Top Fan-in and Variable Set]
	Let $\Phi$ be a multilinear depth-$4$ formula. We say that $\Phi$ is a multilinear $\tfrsSPSP{M}{\tau}$ 
	formula if it is of the form $\sum_{i=1}^m\prod_{j=1}^{t_i} f_{i,j}$, where $m \le M$ and for each $1 \le i \le s$ 
	we have that
	\begin{enumerate}[(i)]
		\item $|\var(f_{i,j})| \le \tau$ for all $1 \le j \le t_i$
		\item $\var(f_{i,j_1}) \cap \var(f_{i,j_2}) = \emptyset$, for any $j_1 \neq j_2$.
	\end{enumerate}
	If, in addition to the conditions above, we have that each $f_{i,j}$ is a linear function, that is, $\Phi$ is actually a 
	depth-$3$ formula, we will say that $\Phi$ is a multilinear $\tfrsSPS{M}{\tau}$ formula.
\end{define}

Since the formula will be given to us as a black-box, we can make some assumptions about it,
which will help us to preserve non-zeroness when taking derivatives or setting variables to zero. To this end, we define a notion of simplicity of depth-$4$ formulas,\footnote{Note that this is not the same notion as used, e.g., in \cite{DvirShpilka06}. } and prove that we can assume without loss of generality that any input formula is simple. 

\begin{define}
	Let $f(x_1, \ldots, x_n) \in \F[x_1, \ldots, x_n]$ be a multilinear polynomial and let 
	$$ \Phi = \sum_{i=1}^M\prod_{j=1}^{t_i} f_{i,j} $$ 
	be a multilinear depth-$4$ formula computing $f$. We say that $\Phi$ is a \emph{simple} multilinear
	depth-$4$ formula if for each variable $x \in \var(f)$ that divides $f$, it must be the case that
	for every $1 \le i \le M$, there exists $j \in [t_i]$ such that $f_{i,j} = x$.
\end{define}

In words, $\Phi$ is simple if whenever a variable $x$ divides $f$, it also divides every product gate. 
The following proposition tells us that we can indeed assume, without loss of generality, that any multilinear depth-$4$ formula given to us
is a simple formula. 

\begin{prop}\label{prop:simple}
	If $\Phi$ is a depth-$4$ multilinear $\tfSPSP{M}$ formula computing $f(x_1, \ldots, x_n)$, then $f$ can be computed 
	by a simple depth-$4$ multilinear $\tfSPSP{M}$ formula $\Psi$ where $|\Psi| \le |\Phi|.$ 
\end{prop}

\begin{proof}
	Since $\Phi$ is a $\tfSPSP{M}$ formula, we have that 
	$$  f = \ds\sum_{i=1}^M\prod_{j=1}^{t_i} f_{i,j}. $$ 
	Let $x \in \var(f)$ be such that $x \mid f$.
	Notice that we can write each $f_{i,j}$ in the following form: 
	$$  f_{i,j} = x \cdot g_{i,j} + h_{i,j}, \ \text{ where } x \not\in \var(g_{i,j}) \cup \var(h_{i,j}). $$
	Moreover, observe that if $x \not\in \var(f_{i,j})$, then we must have that $f_{i,j} = h_{i,j}$. Since the formula is
	multilinear, for each $i \in [M]$ there exists at most one $j$ such that $x \in \var(f_{i,j})$. If such $j$ exists,
	we might as well assume without the loss of generality that $j=1$.
	
	Let $A = \{ i : 1 \le i \le M, \text{ and } x \in \var(f_{i,1})  \}$ and $B = [M] \setminus A$.
	Now, rewriting the formula above for $f$, we get:
	\begin{align*}  
	f = \ds\sum_{i=1}^M\prod_{j=1}^{t_i} f_{i,j} &=  \ds\sum_{i \in A} f_{i,1} \cdot \prod_{j=2}^{t_i} f_{i,j} + 
	\ds\sum_{i \in B}\prod_{j=1}^{t_i} f_{i,j} \\ 
	&=\ds\sum_{i \in A} (xg_{i,1} + h_{i,1}) \cdot \prod_{j=2}^{t_i} h_{i,j}  + \ds\sum_{i \in B}\prod_{j=1}^{t_i} h_{i,j}  \\
	&=\ds\sum_{i \in A} xg_{i,1} \cdot \prod_{j=2}^{t_i} h_{i,j} + \ds\sum_{i \in A} h_{i,1} \cdot \prod_{j=2}^{t_i} h_{i,j} + 
	\ds\sum_{i \in B}\prod_{j=1}^{t_i} h_{i,j}. 
	\end{align*}
	Since $x \mid f$, it follows that $f = xg$. Hence, we must have that (in the above equation)
	$$ \ds\sum_{i \in A} h_{i,1} \cdot \prod_{j=2}^{t_i} h_{i,j} + \ds\sum_{i \in B}\prod_{j=1}^{t_i} h_{i,j} = 0 $$
	and therefore
	$$  f = \ds\sum_{i \in A} xg_{i,1} \cdot \prod_{j=2}^{t_i} h_{i,j}.  $$
	Since $|A| \le M$ and $\| g_{i,1} \| \le \| f_{i,1} \|, \ \| h_{i,j} \| \le \| f_{i,j} \|$ for every $i \in [M]$ and $2 \le j \le k_{i}$, the
	formula 
	\[
	\Phi' = \sum_{i \in A} x \cdot g_{i,1} \cdot \prod_{j=2}^{t_i} h_{i,j} = \sum_{i \in A} \prod_{j=2}^{t_i} x \cdot g_{i,1} \cdot h_{i,j}
	\] is a multilinear $\tfSPSP{M}$ formula
	computing $f$, of size $|\Phi'| \le |\Phi|$ and such that $x$ appears as a polynomial at each product gate.
	
	By repeating this process for each variable $\var(f) \setminus \var^*(f)$, we get our $\tfSPSP{M}$ formula $\Psi$.
	Since at each step we preserve the invariant that the size of the formula does not increase, we must have that
	$|\Psi| \le |\Phi|.$
\end{proof}

As a corollary, together with the simple observation that any derivative or restriction of a multilinear formula results in a multilinear formula of at most the same size, we obtain that 
partial derivatives or restrictions of a multilinear polynomial can also be computed by simple formulas.

\begin{cor}
	If $\Phi$ is a depth-$4$ multilinear $\tfSPSP{M}$ formula computing $f(x_1, \ldots, x_n)$, then for any disjoint sets 
	$A, B \subseteq \var(f)$, $\D_A f|_{B=0}$ can be computed by a simple depth-$4$ multilinear $\tfSPSP{M}$ formula 
	$\Psi$ where $|\Psi| \le |\Phi|.$ We will refer to $\Psi$ as $\D_A \Phi|_{B=0}$.
\end{cor}

Therefore, from now on we will always assume that any depth-$4$ multilinear formula given to us is a simple formula.

\subsection{ROABPs for Products of Sparse Polynomials}

Another important model that we need for our constructions is that of Algebraic Branching Programs. 

\begin{define}[Nisan~\cite{Nisan91}]
	An algebraic branching program (ABP) is a directed acyclic graph with one vertex $s$ of in-degree zero (the {\em source}) and
	one vertex $t$ of out-degree zero (the {\em sink}). The vertices of the graph are partitioned into levels labeled $0, 1, \ldots, D$.
	Edges in the graph can only go from level $\ell-1$ to level $\ell$, for $\ell \in [D]$. The source is the only vertex at level $0$
	and the sink is the only vertex at level $D$. Each edge is labeled with an affine function in the input variables. The {\em width} of an
	ABP is the maximum number of nodes in any layer, and the {\em size} of an ABP is the number of vertices in the ABP.
	
	Each path from $s$ to $t$ computes the polynomial which is the product of the labels of the path edges, and the ABP 
	computes the sum, over all $s\to t$ paths, of such polynomials.  
\end{define}

\begin{define}[Ordered Read-Once Algebraic Branching Programs]
	A {\em Ordered Read-Once Algebraic Branching Program (ROABP)} in the variable set $\{x_1, \ldots, x_D\}$ is an ABP of 
	depth $D$, such that each edge between layer $\ell-1$ and $\ell$ is labeled by a univariate polynomial in $x_\ell$.
\end{define}

In this section we show an elementary construction of ROABPs for a very specific class of polynomials. This construction 
however will be useful in the upcoming sections.

\begin{lem}
\label{lem:roabp-product-sparse-polys}
Let $\F$ be a field, and $f(y_1,\ldots,y_m) = \sum_{i=1}^M \prod_{j=1}^{t_i} f_{i,j}$ be a multivariate polynomial over $\F$, 
such that for every $1 \le i \le M$: 
\begin{enumerate}
\item At most $k$ different  $1 \le j \le t_i$, satisfy $ | \var(f_{i,j}) | > 1$.
\item For every $1\leq j \leq t_i$, $ \| f_{i,j} \| \leq s$.
\end{enumerate}
Then $f$ can be computed by an ROABP of width at most $M \cdot s^{k}$. 
\end{lem}

\begin{proof}
Assume without the loss of generality that for every $i$ there is $k_i\leq k$ such that $f_{i,1},\ldots,f_{i,k_i}$ are those polynomials that contain more than a single variable. 
Note that the sparsity of every product $g_i \eqdef \prod_{j=1}^{k_i} f_{i,j}$ is at most $s^k$. We construct an ROABP of width 
$s^k$ for each $g_i$. The final ROABP is constructed by connecting the ROABPs for all $M$ products in parallel.

Fix $1 \le i \le M$. Expanding the product $g_i = \prod_{j=1}^{k_i} f_{i,j}$ we get at most $s^k$ monomials. Now, multiply each 
such monomial with the remaining functions in the $i$'th gate, $\prod_{j=k_i+1}^{t_i}f_{i,j}$. Notice that the multiplicands in each 
such term can be reordered so that first $x_1$ appears then $x_2$ etc. Thus, we can construct a ROABP of width $s^k$ for 
computing each such product of a monomial with $\prod_{j=k_i+1}^{t_i}f_{i,j}$. Then, connecting all those ROABPs in parallel we get 
a ROABP of width $s^k$ for the $i$'th multiplication gate. Connecting in a similar fashion all ROABPs for the different multiplication 
gates we get a ROABP of width (at most) $M\cdot s^k$ computing $f$. 
\end{proof}

\subsection{Hashing}

In this section we present the basic hashing tools that we will use in our construction. We first recall the notion of a $k$-wise independent hash family.

\begin{define}\label{def:k-hash}
A family of hash functions ${\cal F}=\{ h:[n] \to [m] \}$ is $k$-wise independent if for any $k$ distinct elements $(a_1, \dots, a_k) \in [n]^k$ and any $k$ (not necessarily distinct) elements $(b_1, \dots, b_k) \in [m]^k$, we have:
$$\Pr_{h \in \cF} \left[ h(a_1)=b_1 \wedge \cdots \wedge h(a_k)=b_k \right] = m^{-k}.$$
\end{define}

Our next lemma studies the case where several sets are hashed simultaneously by the same hash function.
We present the proof in a general form and only later, in the application, fix the parameters. 

\begin{lem}
\label{lem:hash}
Let $0<\delta<\eps$, and $n,M \in \N$ such that $M=2^{n^{\delta}}$. Assume $\cA_1,\ldots,\cA_M$ are families of pairwise disjoint subsets of $[n]$
such that for every $1 \le i \le M$ and every $A \in \cA_i$, $|A| \le n^{1-\eps}$.
Let 
$\gamma > 0$ be such that $\gamma \ge (\eps-\delta)/2$.
Let $\cF$ be a family of $k$-wise independent hash functions from $[n]$ to $[m]$ for $k =n^{\delta} + 2\log n $ and $m=10 n^{1-\eps+\gamma}$.

Then there exists $h \in \cF$ such that for every $ 1 \le i \le M$ and every $1 \le j \le m$, both of the following conditions hold:

\begin{enumerate}

\item \label{item:no-intersection-larger-than-k} For every set $A \in \cA_i$, $\left| h^{-1}(j) \cap A \right| \le k$.

\item \label{item:at-most-k-intersect} The number of sets $A \in \cA_i$ such that $\left| h^{-1}(j) \cap A \right| > 1$ is at most $k \log n$.

\end{enumerate}

\end{lem} 

\begin{proof}
We show that for a random $h \in \cF$, both items \ref{item:no-intersection-larger-than-k} and \ref{item:at-most-k-intersect} happen 
with probability at least $2/3$.

Fix $ 1 \le i \le M$, $1 \le j \le m$ and $A \in \cA_i$. By $k$-wise independence and the assumption that $|A| \le n^{1-\eps}$, we have that
\begin{align}
\label{eq:prob-intersect-larger-than-k}
\Pr \left[ \left| h^{-1}(j) \cap A \right| \ge k \right] & \le 
\sum_{B \subseteq A, |B|=k} \Pr [\forall b \in B, h(b)=j] \nonumber \\ 
& \le \binom{n^{1-\eps}}{k} \cdot \frac{1}{ {\left( 10 n^{1-\eps+\gamma} \right)}^{k}}  \nonumber \\ 
& \le n^{(1-\eps)k} \cdot \frac{1}{ {\left( n^{(1-\eps)k + \gamma k} \right)}} \cdot \frac{1}{10^k} \nonumber \\ 
& \le  {n^{-\gamma k}} \cdot 10^{-k}  .
\end{align}

Taking a union bound over all $1 \le i \le M$ and $1 \le j \le m$, and using the estimate \eqref{eq:prob-intersect-larger-than-k} and 
the fact that $m \le n$, we get that item \ref{item:no-intersection-larger-than-k} in the statement of the lemma does not happen 
with probability at most
\[
M \cdot m \cdot n^{-\gamma k} \cdot 10^{-k} \le \frac{1}{3}
\]
for large enough $n$, by the choice of $k$.

Turning to item \ref{item:at-most-k-intersect} in the statement of the lemma, it is convenient to partition every family of subsets $\cA$ 
into $(1-\eps) \log n$ disjoint buckets, according to the size of the sets in $\cA$. Fix such $\cA$, and, for $1 \le i \le (1-\eps) \log n$, 
 define the bucket
\[
\cB_i = \left\{ A \in \cA :  \frac{n^{1-\eps}}{2^{i}} \le |A| \le   \frac{n^{1-\eps}}{2^{i-1}} \right\}.
\]
We show that with high probability over the choice of $h$, and for every $j \in [m]$, in every bucket there are at most $k$ sets whose 
intersection with $h^{-1}(j)$ has size larger than 1.

For every sets $A \in \cA$, by $k$-wise independence (in particular, pairwise independence) the probability that 
$|A \cap h^{-1} (j)| \ge 2$ is at most
\[
\binom{|A|}{2} \cdot \frac{1}{{ \left( 10n^{1-\eps+\gamma} \right)}^2} \le |A|^2 \cdot \frac{1}{100 n^{2-2\eps+2\gamma}}.
\]

Fix a bucket $\cB_i$. By definition, for every  $A \in \cB_i$ it holds that $|A| \le \frac{n^{1-\eps}}{2^{i-1}}$, and so for every set 
$A \in \cB_i$, the probability that $|A \cap h^{-1} (j)| \ge 2$ is at most
\begin{equation}
\label{eq:prob-intersect-2}
\frac{n^{2-2\eps}}{2^{2i-2}} \cdot \frac{1}{100} \cdot n^{2\eps-2-2\gamma} = \frac{1}{100} \cdot {n^{-2\gamma} \cdot 2^{2-2i}}.
\end{equation}

Since $\cA$ is a partition, by pairwise disjointness, the number of sets in $\cB_i$ is at most $n^\eps \cdot 2^{i}$. Hence, by 
$k$-wise independence and \eqref{eq:prob-intersect-2}, the probability there exist $k/2$ sets in the bucket $\cB_i$, with 
intersection sizes at least $2$, is at most

\begin{align}
\label{eq:prob-intersect-k-over-2}
\binom{n^\eps \cdot 2^{i} }{k/2} \cdot {\left( \frac{1}{100} \cdot {n^{-2\gamma} \cdot 2^{2-2i}} \right)}^{k/2} & \le
{\left( \frac{en^\eps \cdot 2^{i}}{k/2} \right)}^{k/2} {\left( \frac{1}{100} \cdot {n^{-2\gamma} \cdot 2^{2-2i}} \right)}^{k/2} \nonumber \\
& = {\left( \frac{e \cdot 2^{3-i}}{100} \right)}^{k/2} \cdot {\left( \frac{n^{\eps-2\gamma}}{k} \right) }^{k/2},
\end{align}
where we have used the inequality $\binom{a}{b} \le \left( \frac{ea}{b} \right) ^b$.  Observe that $n^{\eps-2\gamma} \le k$, by the 
choice of $\gamma$.

Taking a union bound over all $\log n$ buckets, and then over all $M$ partitions and all $m$ possible values of $j$, and using the 
estimation \eqref{eq:prob-intersect-k-over-2}, we get that the probability that  there more than $k/2$ sets that  intersect 
$h^{-1}(j)$ in more than one element, for some $j$, is at most 
\[
M \cdot m \cdot \log n \cdot \left( \frac{e \cdot 2^{3-i}}{100} \right)^{k/2} \cdot \left( \frac{n^{\eps-2\gamma}}{k} \right)^{k/2} \le \frac{1}{3},
\]
for large enough $n$, by the choices of $k$ and $\gamma$. Hence, item \ref{item:at-most-k-intersect} in the statement of the 
lemma follows as well.
\end{proof}

We conclude this section with the following well known fact (see, e.g., Chapter 16 in \cite{AlonSpencer08}, and the references 
therein):

\begin{fact}
\label{fact:small-hash}
There exists an explicitly constructible family $\cF$ of $k$-wise independent hash functions from $[n]$ to 
$[10n^{1-\eps+\gamma}]$ of size $|\cF| = n^{O(k)}$.
\end{fact}

\section{Reducing the Bottom Support of Depth-$3$ and Depth-$4$ Formulas}\label{sec:red}

In this section we make the first step towards proving Theorems~\ref{thm:intro:hitting-set-depth-$3$}  and \ref{thm:intro:hitting-set-depth-$4$}. As outlined in Section~\ref{sec:technique}, our first step is making the functions computed at the bottom layers (linear functions in the case of depth-$3$ and ``sparse'' polynomials in the case of depth-$4$) have small (variable) support. Hence,  we establish reductions from any $\tfSPS{M}$ or $\tfSPSP{M}$ formula to a $\tfrsSPS{M}{\tau}$ or 
$\tfrsSPSP{M}{\tau}$ formula, respectively. We first describe the simple depth-$3$ case. We continue by elaborating on the more general case of depth-$4$ formulas, which is slightly more involved. Both proofs follow the outline described in Section~\ref{sec:technique}.

\subsection{Reducing Bottom Support for Depth-$3$}

Given a depth-$3$ formula $\sum_{i=1}^M \prod_{j=1}^{t_i} \ell_{i,j}$, we would like 
to eliminate all linear functions that contain many variables. To this end, we observe that there must exist a variable that 
appears in many of these functions, and that taking a derivative according to that variable eliminates those functions from the formula.

\begin{lem}
\label{lem:red3}
Let $f(x_1,\ldots,x_n) = \sum_{i=1}^M \prod_{j=1}^{t_i} \ell_{i,j}$ be a non-zero multilinear polynomial computed by a multilinear 
$\tfSPS{M}$ formula $\Phi$ and let $\eps>0$. Then, there 
exists a set of variables $A$ of size $|A| \le \tilde{O} (n^{\epsilon} \cdot \log M)$ such that $\partial_A f$ is a non-zero multilinear 
polynomial that can be computed by a multilinear $\tfrsSPS{M}{n^{1-\eps}}$ formula.
\end{lem}

\begin{proof}
Define
\[
\cB = \{ \ell_{i,j} \mid |\var(\ell_{i,j}) \cap \var(f) | \ge n^{1-\eps} \}
\]
to be the set of ``bad'' linear functions. Those are linear functions that contain more than $n^{1-\eps}$ variables that also 
appear in $f$. We show how to eliminate those linear functions from the formula while preserving non-zeroness.

Since for every $\ell \in \cB$, $|\var(\ell) \cap \var(f)| \ge n^{1-\eps}$, there exists a variable $x_i$ that appears in at least 
$|\cB|n^{1-\eps}/n = |\cB|/n^\eps$ linear functions in $\cB$ (and also in $f$). Consider the polynomial $\partial_{x_i}f$. 
Since $x_i \in \var(f)$, this is a non-zero polynomial. Furthermore, using the fact that deriving with respect to a variable is a 
linear operation, and the fact that every multiplication gate in the formula multiplies linear functions with disjoint support, a formula for 
$\partial_{x_i}f$ can be obtained from $\Phi$ by replacing every linear function in which $x_i$ appears with an appropriate constant. 
Therefore, every such function is removed from $\cB$, and so the set of bad linear functions in $\partial_{x_i}f$ is of size at most 
$|\cB|- |\cB|/n^\eps = |\cB|\cdot (1-1/n^\eps)$. We continue this process for at most $O(n^\eps \cdot \log |\cB|)$ steps, until we 
reach a point where $|\cB| < 1$ and so $|\cB| = 0$.

Finally, it remains to be noted that $|\cB| \le Mn$, since by multilinearity each multiplication gate multiplies linear functions with disjoint 
support, and so its fan-in is at most $n$.
\end{proof}

\subsection{Reducing Bottom Support for Depth-$4$}

The process for depth-$4$ formulas follows the same outline as the depth-$3$ case, but there are a few more details. 
Given a parameter $t \in \N$, we want to transform any multilinear 
$\tfSPSP{M}$ formula computing a non-zero polynomial $f(x_1, \ldots, x_n)$ into a $\tfrsSPSP{M}{\tau}$ formula, 
while preserving non-zeroness. By Proposition~\ref{prop:simple}, we can assume any formula that we work with is 
a \emph{simple formula}. We again define the 
``bad'' polynomials as those that contain many variables (that also appear in $f$). 
Our progress measure for their elimination will be the {\em total sparsity} of all bad polynomials, which we define below.

\begin{define}
	Let $\tau \in \N$ and $\Phi = \sum_{i=1}^M\prod_{j=1}^{t_i} f_{i,j}$ be a multilinear $\tfSPSP{M}$ 
	formula computing a non-zero multilinear polynomial $f(x_1, \ldots, x_n) \in \F[x_1, \ldots, x_n]$. 
	Let $\cB = \{ f_{i,j} \mid \var(f_{i,j}) > \tau \}$. We say that $\Phi$ is $\Delta$-far from a $\tfrsSPSP{M}{\tau}$ formula if
	$$ \sum_{g \in \cB} \| g \| = \Delta.  $$
	We also say that a polynomial $f(x_1, \ldots, x_n)$ is $\Delta$-far from a $\tfrsSPSP{M}{\tau}$ formula if it can be 
	computed by a formula that is $\Delta$-far from such a formula.
\end{define}

\begin{obs}
	Notice that a formula $\Phi$ is $0$-far from being $\tfrsSPSP{M}{\tau}$ iff $\Phi$ itself is a $\tfrsSPSP{M}{\tau}$ formula.
\end{obs}

Now that we have a measure of how far a given $\tfSPSP{M}$ formula (computing a non-zero polynomial) is from being 
$\tfrsSPSP{M}{\tau}$, we 
can show the existence of a small set of variables such that when we either take derivatives or set these variables to zero, 
we obtain a $\tfrsSPSP{M}{\tau}$ formula computing another non-zero polynomial. Since we are working with \emph{simple} formulas, 
if a variable $x$ appears in a bad polynomial $f_{i,j} \in \cB$, then it must be the case that
$x \in \var^*(f)$, and therefore we are free to either take a partial derivative with respect to $x$ or to set $x$ to zero, while preserving
non-zeroness of the input polynomial $f(x_1, \ldots, x_n)$. Therefore, the non-zeroness condition is taken care of by simplicity.

We begin by showing that we can always make good
progress in this measure. More precisely, we have the following lemma:

\begin{lem}\label{lem:red4}
	Let $\Phi = \sum_{i=1}^M\prod_{j=1}^{t_i} f_{i,j}$ be a multilinear 
	$\tfSPSP{M}$ formula computing 
	a non-zero multilinear polynomial $f(x_1, \ldots, x_n) \in \F[x_1, \ldots, x_n]$.  
	If $\Phi$ is $\Delta$-far from a $\tfrsSPSP{M}{\tau}$ formula, then there exists $x \in \var^*(f)$ such that one of the 
	polynomials $\D_xf$ or $f|_{x=0}$ is non-zero and is at most $\Delta(1-\frac{\tau}{2n})$-far from a 
	$\tfrsSPSP{M}{\tau}$ formula.
\end{lem}

\begin{proof}
	By Proposition~\ref{prop:simple}, we can assume without loss of generality that $\Phi$ is simple.
	We note that making $\Phi$ simple can only reduce $\Delta$.

	Let $\cB = \{ f_{i,j} \mid |\var(f_{i,j})| > \tau \}$. Notice that by simplicity of $\Phi$, we have that
	$|\var(f_{i,j})| > 1 \then \var(f_{i,j}) \subseteq \var^*(f)$. Since
	$\Phi$ is $\Delta$-far from a $\tfrsSPSP{M}{\tau}$ formula, we have that
	$$  \Delta = \sum_{g \in \cB} \| g \| .$$
	For each $x \in \var^*(f)$, let $F_x = \{ g \in \cB \mid x \in \var(g) \}$. Notice that 
         
        $$ \ds\sum_{x \in \var^*(f)} \left( \ds\sum_{g \in F_x} \| g \| \right) = \ds\sum_{g \in \cB} |\var(g)| \cdot \| g \| > 
        \tau \cdot \ds\sum_{g \in \cB} \| g \| = \tau \Delta .  $$    
        This implies that there exists $x \in \var^*(f)$ for which 
        
        \begin{equation}\label{eq:dense-var}
        		\ds\sum_{g \in F_x} \| g \| \ge \dfrac{\Delta \cdot \tau}{|\var^*(f)|} \ge \dfrac{\Delta \cdot \tau}{n} .
	\end{equation} 
        Since $ \| g \| = \| g|_{x = 0} \| + \| \D_x g \| $ for any multilinear polynomial $g(x_1, \ldots, x_n)$,  
        we have that  
        $$\ds\sum_{g \in F_x} \| g \|  = \ds\sum_{g \in F_x} \| g|_{x = 0} \| + \ds\sum_{g \in F_x} \| \D_x g \| . $$
        Hence, by equation~\eqref{eq:dense-var}, one of $\ds\sum_{g \in F_x} \| g|_{x = 0} \|$ or $\ds\sum_{g \in F_x} \| \D_x g \|$ 
        must be  larger than $\frac{\Delta \cdot \tau}{2n}$. If  
        $$ \ds\sum_{g \in F_x} \| g|_{x = 0} \| > \frac{\Delta \tau}{2n} ,$$ 
        by taking the derivative of $f$ with respect to $x$, we have that $\D_x f \neq 0$ (since $x \in \var^*(f)$) and that 
        the distance of $\D_x \Phi$ to a $\tfrsSPSP{M}{\tau}$ formula is upper bounded by
        $$  \ds\sum_{g \in \cB \setminus F_x} \| g \| + \ds\sum_{g \in F_x} \| \D_x g \| = \ds\sum_{g \in \cB} \| g \| - 
        \ds\sum_{g \in F_x} \| g|_{x = 0} \| < \Delta - \frac{\Delta \tau}{2n} = \Delta\left(1-\frac{\tau}{2n}\right).$$    
        Notice that the bound above is an upper bound, since the new set of polynomials $f_{i,j}$ of large support must be a subset of
        $\cB.$
        
        Analogously, if $\sum_{C \in F_x} \| \D_x C \| > \frac{\Delta \tau}{2n}$, then we take $f|_{x=0}$. The upper bounds are the same as those obtained for the first case. This finishes the proof of the lemma.
\end{proof}

By repeatedly applying Lemma~\ref{lem:red4}, we obtain the following corollary, which guarantees the existence of a small set
of variables that allow us to transform our $\tfSPSP{M}$ formula into a $\tfrsSPSP{M}{\tau}$ one.

\sloppy
\begin{cor}[Reduction to Depth-$4$ with Small Bottom Support]\label{cor:depth4}
	Let $\Phi$ be a multilinear simple $\tfSPSP{M}$ formula computing a non-zero multilinear 
	polynomial $f(x_1, \ldots, x_n) \in \F[x_1, \ldots, x_n]$. There exist disjoint sets $A, B \subset [n]$ with 
	$|A \sqcup B| \le \frac{2n}{\tau} \cdot \log(|\Phi|)$ such that the polynomial $\D_Af|_{B=0}$ is non-zero and 
	can be computed by a simple multilinear $\tfrsSPSP{M}{\tau}$ formula $\Phi$.
\end{cor}

\begin{proof}
	Let $\Delta$ be such that $\Phi$ is $\Delta$-far from being $\tfrsSPSP{M}{\tau}$. Notice that $\Delta \le |\Phi|$.
	
	We show by induction that there exist disjoint sets $A_k$ and $B_k$ such that 
	$|A_k \sqcup B_k| \le k$, and the polynomial 
	$\D_{A_k} f|_{B_k=0}$ is non-zero and at most $\Delta(1-\frac{\tau}{2n})^k$-far from being $\tfrsSPSP{M}{\tau}$.
	
	For $k \ge 0$, define $A_k, B_k \subseteq [n]$, $f_k(x_1, \ldots, x_n) = \D_{A_k} f|_{B_k=0}$ and $\Delta_k$ be an upper 
	bound on how far
	$f_k(x_1, \ldots, x_n)$ is from being $\tfrsSPSP{M}{\tau}$. Initially, set $A_0 = B_0 = \emptyset$. In this case, we have 
	that $f_0(x_1, \ldots, x_n) = f(x_1, \ldots, x_n)$
	and $\Delta_0 = \Delta = \Delta(1-\frac{\tau}{2n})^0$. This is the base case for our induction.
	
	Suppose our hypothesis is true for some $k \ge 0$. If $\Delta(1-\frac{\tau}{2n})^k < 1$, then we know that our formula is already
	$\tfrsSPSP{M}{\tau}$ and therefore we are done. Else, by applying Lemma~\ref{lem:red4}, we have that there is a variable
	$x \in \var^*(f_k)$ such that either $\D_x f_k$ or $f_k|_{x=0}$ is (at most) $\Delta_k(1-\frac{\tau}{2n})$-far from being $\tfrsSPSP{M}{\tau}$.
	Thus, $\Delta_k(1-\frac{\tau}{2n}) \le \Delta(1-\frac{\tau}{2n})^k \cdot (1-\frac{\tau}{2n}) = \Delta(1-\frac{\tau}{2n})^{k+1}$ and 
	$x \in \var^*(f_k) \subseteq [n] \setminus (A_k \sqcup B_k)$. Therefore, if $\D_x f_k$ is the close polynomial then  
	we set $A_{k+1} = A_k \cup \{x\}, B_{k+1} = B_k$. Otherwise, we set $A_{k+1} = A_k, B_{k+1} = B_k \cup \{x\}.$ It is 
	easy to see that the inductive properties hold in this case as well. This ends the inductive proof.
	
	Since $\Delta(1-\frac{\tau}{2n})^k < 1$ for $k \ge \frac{2n}{\tau} \log \Delta$, and since 
	$\frac{2n}{\tau} \log(|\Phi|) \ge \frac{2n}{\tau} \log \Delta$,
	it is enough to choose at most $\frac{2n}{\tau} \log(|\Phi|)$ variables. This proves this corollary.
\end{proof}

\section{Hitting Set for $\rsSPS{n^{1-\epsilon}}$ and $\rsSPSP{n^{1-\epsilon}}$ Formulas}
\label{sec:hit-bottom}

In this section we construct subexponential sized hitting set for the classes of $\tfrsSPS{M}{n^{1-\eps}}$ and $\tfrsSPSP{M}{n^{1-\eps}}$ multilinear formulas. Recall that in Section~\ref{sec:red} we showed how to reduce general depth-$3$ and depth-$4$ formulas to these types of formulas. In the next section, we will show how to tie all loose edges and combine the arguments of Section~\ref{sec:red} with those of this section in order to handle the general case.

An essential ingredient in our construction is a quasi-polynomial sized hitting set for Read-Once Algebraic Branching 
Programs (ROABPs) \cite{ForbesShpilka13,AgrawalGKS14}. We note that in our setting, we may assume that the reading order of the variables by the ABP is known.

\begin{thm}[\cite{ForbesShpilka13,AgrawalGKS14}]
\label{thm:hitting-set-ROABP}
Let $\cC$ be the class of $n$-variate polynomials computed by a ROABP of width $w$, such that the degree of each variable is at most $d$, over a field $\F$ so that $|\F| \ge \poly(n,w,d)$. Then $\cC$ has a hitting set of size $\poly(n,w,d)^{\log n}$ that can be constructed in time $\poly(n,w,d)^{\log n}$.
\end{thm}

We begin by describing a unified construction for both $\tfrsSPS{M}{n^{1-\eps}}$ and $\tfrsSPSP{M}{n^{1-\eps}}$ formulas. 
We then describe how to set the parameters of the construction for each of the cases.

\begin{construct}[Hitting set for multilinear $\tfrsSPS{M}{n^{1-\eps}}$ and $\tfrsSPSP{M}{n^{1-\eps}}$ formulas]
\label{con:subexp-hitting-set-gen-multilinear}
Let $0<\delta<\epsilon$ and $n,k,s,M$ integers, such that $M=2^{n^\delta}$ and  $k= n^{\delta} + 2\log n$. 
Set 
$m=10n^{1-(\eps+\delta)/2}$ and $t=k \log n$.
Let $\cF$ be a family of $k$-wise independent hash functions from $[n]$ to $[m]$, as in Lemma~\ref{lem:hash}.
For every $h \in \cF$, define the set $I_h$ as follows:
\begin{enumerate}
\item Partition the variables to sets\footnote{Recall that we associate subsets of $[n]$ with subsets of the variables, and make no distinction in the notation.} $T_1 \sqcup T_2 \sqcup \cdots \sqcup T_m = h^{-1}(1) \sqcup h^{-1}(2) \sqcup \cdots \sqcup h^{-1}(m)$.
\item For every $1 \le i\le m$, let $\cH_i$ be a hitting set for ROABPs of width $M \cdot s^t$ and individual degree $d=1$ (as promised by Theorem~\ref{thm:hitting-set-ROABP}), on the variables of $T_i$ (of course, $|T_i|\leq n$).
\item We define $I_h$ as the set of all vectors $v$ such that the restriction of $v$ to the coordinates $T_i$, $v|_{T_i}$, is in $\cH_i$. I.e., in the notation of Section~\ref{sec:notation}, 
$$I_h = \cH_1^{T_1} \times \cH_2^{T_2} \times \cdots \times \cH_m^{T_m}.$$ 
\end{enumerate}
Finally, define $\cH = \bigcup_{h\in \cF} I_h$.
\end{construct}

The following lemma gives an upper bound to the size of the hitting set constructed in Construction~\ref{con:subexp-hitting-set-gen-multilinear}. 

\begin{lem}
\label{lem:subexp-hitting-set-gen-size}
Let $\delta,\epsilon,k,n,s$ and $M$ be the parameters of Construction \ref{con:subexp-hitting-set-gen-multilinear}. 
The set $\cH$ constructed in Construction \ref{con:subexp-hitting-set-gen-multilinear} has size 
$n^{O(k)} \cdot {\left (M \cdot s^{k \log n} \right)}^{\tilde{O}(n^{1-(\eps+\delta)/2})}={\left (M \cdot s^{k \log n} \right)}^{\tilde{O}(n^{1-(\eps+\delta)/2})}$, and it can be constructed in time $\poly(|\cH|)$.
\end{lem}

\begin{proof}
This is a direct consequence of the construction, Fact~\ref{fact:small-hash} and Theorem~\ref{thm:hitting-set-ROABP}.
\end{proof}

\subsection{Depth-$3$ Formulas}

We begin by describing the argument for depth-$3$ formulas. The following lemma proves that indeed, by setting the proper parameters, the set $\cH$ from Construction \ref{con:subexp-hitting-set-gen-multilinear} does hit $\tfrsSPS{M}{n^{1-\eps}}$ formulas.

\begin{lem}
\label{lem:depth-$3$-multilinear-hitting-set-hits-small-support}
Let $f(x_1,\ldots, x_n) \in \F[x_1,\ldots,x_n]$ be a multilinear polynomial computed by a multilinear $\tfrsSPS{M}{n^{1-\eps}}$ formula  
$\Phi = \sum_{i=1}^M \prod_{j=1}^{t_i} \ell_{i,j} $. Let $\cH$ be the set constructed in 
Construction~\ref{con:subexp-hitting-set-gen-multilinear} with $s=k+1$. Then there exists a point $\valpha \in \cH$ such that 
$f(\valpha) \neq 0$.
\end{lem}

\begin{proof}
For every multiplication gate $1 \le i \le M$ in $\Phi$, define a partition of the variables
\[
\cA_i = \{ \var(\ell_{i,j}) \cap \var(f) \mid 1 \le j \le t_i \}.
\]
Let $h \in \cF$ be the function guaranteed by Lemma~\ref{lem:hash} with respect to the partition $\cA_1,\ldots\cA_M$, and assume the setup of Construction~\ref{con:subexp-hitting-set-gen-multilinear}. We claim that there exists $\valpha \in I_h$ such that $f(\valpha) \neq 0$.

To that end, consider the partition of the variables induced by $h$: 
\[
T_1 \sqcup T_2 \sqcup \cdots \sqcup T_m = h^{-1}(1) \sqcup h^{-1}(2) \sqcup \cdots \sqcup h^{-1}(m).
\]

We view the polynomial as a polynomial $f_1$ in the variables of $T_1$, over the extension field $\F(T_2 \sqcup \cdots \sqcup T_n)$. We claim that $f_1$ can be computed by an ROABP of width $M \cdot {(k+1)}^{k \log n}$. To see this note that, by Lemma~\ref{lem:hash}, in any multiplication gate, at most $k \log n$ linear functions contain more than one variable from $T_1$, and each contains at most $k$ variables. It follows that the sparsity of every linear function (with respect to the variables in $T_1$) among those $k \log n$ functions, is at most $k+1$. By Lemma~\ref{lem:roabp-product-sparse-polys}, $f_1$ can be computed by an ROABP over $\F(T_2 \sqcup \cdots \sqcup T_n)$ of width $M \cdot {(k+1)}^{k \log n}$. By the hitting set property of Theorem~\ref{thm:hitting-set-ROABP},
  there exists $\ol{\alpha_1} \in \cH_1 \subseteq \F^{|T_1|}$ such that $f_2 \eqdef f_1|_{T_1 = \ol{\alpha_1}} \not\equiv 0$.

Similarly, the same conditions now hold for $f_2$, considered as a polynomial over the field $\F(T_3 \sqcup \cdots \sqcup T_n)$, and so there exists $\ol{\alpha_2} \in \cH_2 \subseteq \F^{|T_2|}$ such that $f_3 \eqdef f_2|_{T_2 = \ol{\alpha_2}} \not\equiv 0$.

We continue this process for $m$ steps, and at the last step we find $\ol{\alpha_{m}}$ such that $f_{m-1}(\ol{\alpha_{m}}) = f(\ol{\alpha_1},\cdots,\ol{\alpha_m}) \neq 0$, with $(\ol{\alpha_1},\cdots,\ol{\alpha_m}) \in \F^n$ being the length $n$ vector obtained by filling the entires of $\ol{\alpha_i} \in \F^{|T_i|}$ in the positions indexed by $T_i$.
\end{proof}

\subsection{Depth-$4$ Formulas}

The argument for depth-$4$ formulas is very similar, apart from a small change in the setting of the parameters.

\begin{lem}
\label{lem:depth-$4$-multilinear-hitting-set-hits-small-support}
Let $f(x_1,\ldots, x_n) \in \F[x_1,\ldots,x_n]$ be a multilinear polynomial computed by a multilinear $\tfrsSPSP{M}{n^{1-\eps}}$ formula  
$\Phi = \sum_{i=1}^M \prod_{j=1}^{t_i} f_{i,j} $. Let $\cH$ be the set constructed in 
Lemma~\ref{con:subexp-hitting-set-gen-multilinear} with $s=2^k$. Then, there exists a point $\valpha \in \cH$ such that 
$f(\valpha) \neq 0$.
\end{lem}

\begin{proof}
The proof is almost identical to that of Lemma~\ref{lem:depth-$3$-multilinear-hitting-set-hits-small-support}. In this case, for every $1 \le i \le M$ we define the partition
\[
\cA_i = \{ \var(f_{i,j}) \mid 1 \le j \le t_i \}.
\]

Note that the assumptions of Lemma~\ref{lem:hash} still hold, and so we denote by $h \in \cF$ the function guaranteed by Lemma~\ref{lem:hash} with respect to the partitions $\cA_1,\ldots\cA_M$, and again claim that there exists $\valpha \in I_h$ such that $f(\valpha) \neq 0$.

Consider once more the partition on the variables induced by $h$, 
\[
T_1 \sqcup T_2 \sqcup \cdots \sqcup T_m = h^{-1}(1) \sqcup h^{-1}(2) \sqcup \cdots \sqcup h^{-1}(m),
\]
and view the polynomial as a polynomial $f_1$ in the variables of $T_1$, over the field $\F(T_2 \sqcup \cdots \sqcup T_n)$. 

We now claim that $f_1$ can be computed by an ROABP of width $M \cdot {\left( 2^k \right)}^{k \log n} = 2^{k^2 \log n}$. The proof for this claim is exactly as in the depth-$3$ case, except that now the best bound we can give on the sparsity of each polynomial which intersects $T_1$ in more than one variable is $2^k$, as it is a multilinear polynomials in at most $k$ variables.

Similarly, we move on to handle $T_2,\ldots,T_m$ and obtain a point $\valpha$ such that $f(\valpha) \neq 0$.
\end{proof}

\section{Hitting Set for Depth-$3$ and Depth-$4$ Multilinear Formulas}
\label{sec:hit}

Recall that, in Section~\ref{sec:red}, we showed that any non-zero $\tfSPS{M}$ or $\tfSPSP{M}$ formula has a non-zero partial derivative (and, possibly, a restriction) which is computed by a non-zero $\tfrsSPS{M}{n^{1-\eps}}$ or $\tfrsSPSP{M}{n^{1-\eps}}$ formula, respectively. Then, in Section~\ref{sec:hit-bottom} we gave hitting sets for such formulas. In this section we provide the final ingredient, which is to show how to ``lift'' those hitting sets back to the general class, via a simple method, albeit one that requires some notation.

Handling restrictions is fairly easy, and causes no blow up in the hitting set size: If we have a set $\cH \subseteq \F^{n-r}$ that hits $f|_{B=0}$ for some multilinear polynomial $f(x_1,\ldots,x_n)$ and $B \subseteq [n]$ with $|B|=r$, then simply extending $\cH$ into a subset of $\F^n$ by filling out zeros in all the entries indexed by $B$ will hit $f$ itself.

In order to handle partial derivates, first note that if $f(x_1,\ldots,x_n)$ is a multilinear polynomial, then
\[
\partial_{x_i} f = f(x_1,\ldots,x_{i-1},1,x_{i+1},\ldots,x_n) - f(x_1,\ldots,x_{i-1},0,x_{i+1},\ldots,x_n),
\]
and so if $\partial_{x_i}f (\valpha) \neq 0$ for some $\valpha \in \F^n$ then at least one of the two evaluations on the right hand side must be non-zero as well.

Applying this fact repeatedly, given a set $A \subseteq [n]$ we can evaluate $\partial_A f$ at any point by making $2^{|A|}$ evaluations of $f$. Motivated by this, we introduce the following notation:

\begin{define}
\label{def:extension}
Let $f(x_1,\ldots,x_n) \in \F[x_1,\ldots,x_n]$ be a multilinear polynomial and $A,B \subseteq [n]$ be two disjoint subsets of variables with $|A|=r, |B|=r'$. Let $\cH \subseteq \F^{n-(r+r')}$.

We define the ``lift'' of $\cH$ with respect to $(A,B)$ to be $$\Ext_A^B(\cH) = \left(\{0,1\}^{\vphantom{r'}r}\right)^A \times \left(\{0\}^{r'}\right)^B \times \cH^{[n]\setminus  (A\sqcup B)}.$$ In the special case where $B=\emptyset$, we simply denote $\Ext_A^B(\cH) = \Ext_A (\cH)$.

%
\end{define}

 That is, for all $\valpha\in\cH$, $\Ext_A^B(\cH)$ contains all the possible $2^r$ ways to extend $\valpha$ into $\vbeta \in \F^n$ by filling out zeros and ones within the $r$ entries that are indexed by $A$, and zeros in all the $r'$ entries indexed by $B$.

\subsection{Depth-$3$ Formulas}

In this section we prove Theorem~\ref{thm:intro:hitting-set-depth-$3$}. For the reader's convenience, we first restate the theorem:

\begin{thm}[Theorem~\ref{thm:intro:hitting-set-depth-$3$}, restated]
\label{thm:hitting-set-depth-$3$}
Let $\cC$ be the class of multilinear $\tfSPS{M}$ formulas for $M=2^{n^\delta}$.
There exists a hitting set $\cH$ of size $|\cH| = 2^{\tilde{O}(n^{2/3+2\delta/3})}$ for $\cC$, that can be constructed in time $\poly(|\cH|)$.
\end{thm}

The size of the hitting set is subexponential for any constant $\delta < 1/2$. Also, if $M=\poly(n)$ then the size of the hitting set is $2^{\tilde{O}(n^{2/3})}$.

With Definition~\ref{def:extension} in hand, we now provide our construction for $\tfSPS{M}$ formulas, towards the proof of Theorem~\ref{thm:hitting-set-depth-$3$}.

\begin{construct}[Hitting set for multilinear $\tfSPS{M}$ formulas]
\label{con:hit3}
Let $M=2^{n^\delta}$ and  $\eps =2/3 - \delta/3$. Let $r = \tilde{O}(n^{\eps} \log M) = \tilde{O}(n^{\frac{2}{3} + \frac{2}{3}\delta})$ as guaranteed by Lemma~\ref{lem:red3}.
For every $A \in \binom{[n]}{\le r}$, construct a set $\cH_A \in \F^{n-|A|}$ using 
Construction~\ref{con:subexp-hitting-set-gen-multilinear} with parameters $\delta,\epsilon,n,k,s=k+1$ and $M$ (recall that  in Construction~\ref{con:hit3} we set $k=n^{\delta} + 2 \log n$).
Finally, let
\[
\cH = \bigcup_{A \in \binom{[n]}{\le r}} \Ext_A(\cH_A).
\]
\end{construct}

We are now ready to prove Theorem~\ref{thm:hitting-set-depth-$3$}:

\begin{proof}[Proof of Theorem \ref{thm:hitting-set-depth-$3$}]
We show that the set $\cH$ constructed in Construction~\ref{con:hit3} has the desired properties.
First, note that by Lemma~\ref{lem:subexp-hitting-set-gen-size}, for every $A \subseteq [n]$ with $$|A| \leq  \tilde{O}(n^{\eps} \log M) =\tilde{O}(n^{2/3-\delta/3} \log M)= \tilde{O}(n^{\frac{2}{3} + \frac{2}{3}\delta}),$$
the set $\cH_A$ has size
\[
(M \cdot (k+1)^{k \log n})^{\tilde{O}(n^{2/3-\delta/3})} = 2^{\tilde{O}(n^{2/3+2\delta/3})},
\]
where we have used the fact that, in Construction~\ref{con:hit3}, we take $k=n^{\delta} + 2 \log n$.
It therefore follows that 
\[
|\Ext_A(\cH_A)| \le 2^{|A|} \cdot |\cH_A| = 2^{\tilde{O}(n^{2/3+2\delta/3})},
\]
and that
\[
|\cH| \le \sum_{i=0}^{\tilde{O}(n^{\frac{2}{3} + \frac{2}{3}\delta})} \sum_{A \subseteq [n], |A| = i} |\Ext_A(\cH_A)| =  2^{\tilde{O}(n^{2/3+2\delta/3})}.
\]

To show the hitting property of $\cH$, let $f(x_1,\ldots,x_n)$ be a non-zero multilinear polynomial computed by a $\tfSPS{M}$ formula, and let $A' \subseteq [n]$ be the set guaranteed by Lemma~\ref{lem:red3}.  Thus, $|A'| \leq \tilde{O}(n^{\eps} \log M)= \tilde{O}(n^{\frac{2}{3} + \frac{2}{3}\delta})$. Then by Lemma~\ref{lem:depth-$3$-multilinear-hitting-set-hits-small-support}, 
there exists $\alpha \in \cH_{A'}$ such that $\partial_{A'} f (\valpha) \neq 0$, and so there must exist
\[
\vbeta \in \Ext_{A'}(\cH_{A'}) \subseteq \cH
\] 
such that $f(\vbeta) \neq 0$.
\end{proof}

\subsection{Depth-$4$ Formulas}

Moving on to depth-$4$, the construction and proof are both very similar, with a slight change in the parameters. We first give a slightly more general form of Theorem~\ref{thm:intro:hitting-set-depth-$4$} that we will later use for regular formulas.

\begin{thm}[General theorem for multilinear $\SPSP$ formulas]
\label{thm:hitting-set-depth-$4$-general}
Let $\cC$ be the class of multilinear $\SPSP$ formulas of top fan-in $M$ and size $S$ so that $(\log M)^3\cdot \log S = o(n)$.
There exists a hitting set $\cH$ of size $|\cH| =  2^{\tilde{O}( n^{2/3} \cdot \log M \cdot (\log S)^{1/3} )}$ for $\cC$, that can be constructed in time $\poly(|\cH|)$.
\end{thm}

We note that Theorem~\ref{thm:intro:hitting-set-depth-$4$} is an immediate corollary of Theorem~\ref{thm:hitting-set-depth-$4$-general}.

\begin{proof}[Proof of Theorem~\ref{thm:intro:hitting-set-depth-$4$}]
Apply Theorem~\ref{thm:hitting-set-depth-$4$-general} with $M=S=2^{n^{\delta}}$ for some constant $0<\delta<1/4$ (the bound in Theorem~\ref{thm:intro:hitting-set-depth-$4$} is meaningless for $\delta\geq 1/4$). It is clear that the conditions of Theorem~\ref{thm:hitting-set-depth-$4$-general} are met. Thus, we obtain a hitting set of size $|\cH| =  2^{\tilde{O}( n^{2/3} \cdot \log M \cdot (\log S)^{1/3} )} = 2^{\tilde{O}( n^{2/3 + 4\delta/3})}$ for the class.  
\end{proof}

For the proof of Theorem~\ref{thm:hitting-set-depth-$4$-general} we will use the following construction that is similar to Construction~\ref{con:hit3}.

\begin{construct}[Hitting set for multilinear $\SPSP$ formulas]
\label{con:hit4-general}
Let $M$ and $S$ be such that $(\log M)^3\cdot \log S = o(n)$.
Denote $M=2^{n^\delta}$ (hence $S = 2^{o(n^{1-3\delta})}$).
Let $\epsilon$ be such that $$n^\epsilon = n^{2/3}\cdot \log M/(\log S)^{2/3}.$$ Set $$r =2n^{\eps} \log S = 2n^{2/3} \cdot \log M \cdot (\log S)^{1/3}.$$

For every two disjoint sets $A,B \subseteq [n]$ with $|A|,|B|\le r$, construct a set $\cH_{A,B} \in \F^{n-(|A|+|B|)}$ 
using Construction~\ref{con:subexp-hitting-set-gen-multilinear} with parameters $\delta,\epsilon,n,k,s=2^k$ and $M$ (recall that in Construction~\ref{con:subexp-hitting-set-gen-multilinear} we set $k=n^\delta + 2 \log n$).
Finally, let
\[
\cH = \bigcup_{\substack{A,B \in \binom{[n]}{\le r} \\ A\cap B=\emptyset}} \Ext_A^B(\cH_{A,B}).
\]
\end{construct}

\begin{proof}[Proof of Theorem~\ref{thm:hitting-set-depth-$4$-general}]
We show that the set $\cH$ constructed in Construction~\ref{con:hit4-general} has the desired properties. 
First, note that by Lemma~\ref{lem:subexp-hitting-set-gen-size}, for every $A, B \subseteq [n]$ with 
$|A|, |B|  \le 2n^{\epsilon} \log S$, the set $\cH_{A,B}$ has size
\[{\left (M \cdot 2^{k^2 \log n} \right)}^{\tilde{O}(n^{1-(\eps+\delta)/2})} = 
{\left ( 2^{k^2 \log n} \right)}^{\tilde{O}(n^{1-(\eps+\delta)/2})}=2^{\tilde{O}(n^{1-\eps/2 +3\delta/2})} ,
\]
for $k=n^\delta + 2 \log n$.
It therefore follows that 
\[
|\Ext_A^B(\cH_{A,B})| \le 2^{|A|} \cdot |\cH_{A,B}| = 2^{2n^{\epsilon} \log S} \cdot 2^{\tilde{O}(n^{1-\eps/2 +3\delta/2})},
\]
and that
\[
|\cH| \le \sum_{i,j=0}^{2n^{\epsilon} \log S} 
\sum_{\substack{A \subseteq [n] \\ |A| = i}}
\sum_{\substack{B \subseteq [n] \\ |B| = j}} |\Ext_A^B(\cH_{A,B})| =  2^{\tilde{O}(n^{\epsilon} \log S)} \cdot 2^{\tilde{O}(n^{1-\eps/2 +3\delta/2})}.
\]
By our setting of parameters
\begin{align*}
|\cH| \le 2^{\tilde{O}(n^{\epsilon} \log S)} \cdot 2^{\tilde{O}(n^{1-\eps/2 +3\delta/2})} &= 2^{\tilde{O}(n^{\epsilon} \log S + n^{1-\eps/2 +3\delta/2})} \\ &=^{(\dagger)} 2^{\tilde{O}( n^{2/3} \cdot \log M \cdot (\log S)^{1/3} )} ,
\end{align*}
where equality $(\dagger)$ follows from our choice of $\epsilon$ in Construction~\ref{con:hit4-general} and the fact that $\log M = n^\delta$.

To show the hitting property of $\cH$, let $f(x_1,\ldots,x_n)$ be a non-zero multilinear polynomial computed by a multilinear $\tfSPSP{M}$ formula of size $S$, and let $A',B' \subseteq [n]$ be 
the sets guaranteed by Lemma~\ref{cor:depth4} with $\tau=n^{1-\epsilon}$. Thus, $|A'|,|B'|\leq 2n^{\epsilon} \log S$. Then, by 
Corollary~\ref{lem:depth-$4$-multilinear-hitting-set-hits-small-support}, there exists $\valpha \in \cH_{A',B'}$ such that 
$\partial_{A'} f|_{B'=0} (\valpha) \neq 0$, and so there must exist
\[
\vbeta \in \Ext_{A'}^{B'}(\cH_{A',B'}) \subseteq \cH
\] 
such that $f(\vbeta) \neq 0$.
\end{proof}

\section{Multilinear Depth-$d$ Regular Formulas}
\label{sec:regular}

\subsection{Definition and Background}

In \cite{KayalSS14}, Kayal et al.\ define \emph{regular formulas}, which consist of formulas with alternating layers of sum
and product gates such that the fan-in of all the gates in any fixed layer is the same. In addition, they require the formal (syntactic) degree of the formula must be at most twice the (total) degree of its output polynomial.
They showed that any $n^{O(1)}$-sized arithmetic circuit can be computed by a regular formula of size $n^{O(\log^2 n)}$ and proved a lower bound of $n^{\Omega(\log n)}$ on the size of regular formulas that compute some explicit polynomial in $\VNP$. 

In this paper, we consider \emph{multilinear regular formulas}, which are regular formulas with the extra condition that each gate
computes a multilinear polynomial. However, we will remove the bound on the formal degree of the formula. 
More precisely, we have the following definition:

\begin{define}[Multilinear Regular Formulas]\label{def:reg-form}
	We say that a formula $\Phi$ is a multilinear $(a_1, p_1, a_2, p_2, \ldots, a_d, p_d, a_{d+1})$-regular formula computing a
	multilinear polynomial $f(x_1, \ldots, x_n)$ if it can be computed by a multilinear
	$ \Sigma^{[a_1]}\Pi^{[p_1]}\Sigma^{[a_2]}\Pi^{[p_2]}\ldots\Sigma^{[a_d]}\Pi^{[p_d]}\Sigma^{[a_{d+1}]}\text{-formula} $.
	Notice that the size of such a formula is $(\prod_{1 \le i \le d+1} a_i) \cdot (\prod_{1 \le i \le d} p_i)$ and the formal
	degree of such a formula is given by $\deg(\Phi) = \prod_{1 \le i \le d} p_i$. Since the formula is multilinear, we have that
	$\deg(\Phi) \le n$.
\end{define}

Comparing with the definition given in Section~\ref{sec:def}, an $(a_1, p_1, a_2, p_2, \ldots, a_d, p_d, a_{d+1})$-regular formula has depth $2d+1$. 

\subsection{Reduction to Depth-$4$ Formulas}

In this section, we reduce a multilinear depth-$d$ regular formula to a depth-$4$ formula. 
We first give a depth reduction lemma (Lemma~\ref{lem:squeeze}) that tells us that we can reduce the depth by one with a slight blow up in the fan-ins of the regular formula. We then use this lemma to obtain a depth-$4$ formula. The idea is to break the regular formula into two formulas (the top part and the bottom part), and then to apply the depth reduction lemma separately to these two formulas. The delicate part is that we wish to obtain a depth-$4$ formula that has a subexponential size hitting set, as in Theorem~\ref{thm:hitting-set-depth-$4$-general}. For this we need the top fan-in $M$ and the total size $S$ to satisfy that $(\log M)^3\cdot \log S  =o(n)$. To achieve this we should carefully select the point in which to divide the formula. This is done in Theorem~\ref{thm:squeeze}.


We start with the depth reduction lemma.

\begin{lem}[Depth Reduction Lemma]\label{lem:squeeze}
	Let $\Psi$ be a multilinear $(a_1, p_1, a_2, p_2, 1)$-regular formula computing a polynomial $f(x_1, \ldots, x_n)$.
	Then, there exists a multilinear $(a_1a_2^{p_1}, p_1p_2, 1)$-regular formula $\Phi$ computing $f(x_1, \ldots, x_n)$.
\end{lem}

\begin{proof}
	Notice that a multilinear $(a_1, p_1, a_2, p_2, 1)$-regular formula is a $\Sigma^{[a_1]}\Pi^{[p_1]}\Sigma^{[a_2]}\Pi^{[p_2]}$
	formula. Writing $\Psi$ as $\sum_{i=1}^{a_1}\prod_{j_i=1}^{p_1}\sum_{k_{j_i}=1}^{a_2} m(i, j_{i}, k_{j_{i}})$, where each 
	$m(i, j_{i}, k_{j_{i}})$ is
	a monomial that is a product of $p_2$ input gates, and by expanding the expression above by computing all the products,
	we obtain:
	\begin{equation}
	\label{eq:regular-squeeze}
 \ds\sum_{i=1}^{a_1}\prod_{j_i=1}^{p_1}\sum_{k_{j_i}=1}^{a_2} m(i, j_{i}, k_{j_{i}}) =  
	\ds\sum_{i=1}^{a_1}\left(\sum_{k_{i,1}=1}^{a_2}\sum_{k_{i,2}=1}^{a_2} \ldots \sum_{k_{i,p_1}=1}^{a_2} 
	\prod_{t=1}^{p_1} m(i, t , k_{i,t}) \right). 
	\end{equation}
 
	Since $m(i, j_{i}, k_{j_{i}})$ is a product of $p_2$ input gates, and since the right hand side of \eqref{eq:regular-squeeze} computes a product of $p_1$ of these
	terms, each monomial computed by $\prod_{t=1}^{p_1} m(i, t , k_{i,t})$ is a product of $p_1p_2$ input gates.
	
	Since the sums on the right hand side of \eqref{eq:regular-squeeze} are over all tuples of the form $(i, k_{i,1}, k_{i,2} \ldots, k_{i,p_1}) \in [a_1] \times [a_2]^{p_1}$,
	we have that there are exactly $a_1 \cdot a_2^{p_1}$ summands. Hence, the right hand side of \eqref{eq:regular-squeeze} is the expression of a multilinear
	$(a_1a_2^{p_1}, p_1p_2, 1)$-regular formula.
\end{proof}

By repeatedly applying the depth reduction lemma above, we obtain the following theorem:

\begin{thm}[Depth Reduction of Regular Formulas]\label{thm:squeeze}
	Let $d \ge 2$ be an integer, $c \in \R$ a constant such that $c \ge 3$, and $\Psi$ a multilinear 
	$(a_1, p_1, a_2, p_2, \ldots, a_d, p_d, a_{d+1})$-regular formula 
	of size $S$ computing a multilinear polynomial $f(x_1, \ldots, x_n)$. Then, one of the following conditions
	must happen:
	\begin{enumerate}[(i)]
		\item \label{item:bounded-size} For $M=S$, there exists a $\tfSPSP{M}$ formula of size $O(S \cdot n^{n^{1-(1/c)^d}})$ computing 
		$f(x_1, \ldots, x_n)$, or
	\item \label{item:bounded-fan-in} There exists $t \in [d-1]$ such that  
		there is a multilinear $\tfSPSP{M}$ formula $\Phi$ computing $f(x_1, \ldots, x_n)$, 
		with 
		$$ M = S^{n^\alpha} \ \text{ and } \ |\Phi| \le 2Mn \cdot n^{n^{1-(c-1)\alpha}} , $$
		for $ \alpha = \frac{1}{c-1} \cdot \left(\frac{1}{c}\right)^{d-t}$.
	\end{enumerate}
\end{thm}

\begin{proof}
	Recall that the size of the formula $S$ satisfies $S=(\prod_{1 \le i \le d+1} a_i) \cdot (\prod_{1 \le i \le d} p_i)$ .
	We have three cases to analyze:
	
	\paragraph{Case 1 (small total degree):} If $\prod_{i=1}^d p_i \le n^{1-(1/c)^d}$, then we can simply write the polynomial 
	$f(x_1, \ldots, x_n)$ as a sum of monomials, which would give us a multilinear $\Sigma \Pi$ formula $\Phi$ 
	of size
	$$ |\Phi| \le \ds\sum_{i=0}^{n^{1-(1/c)^d}} \binom{n}{i} = O(n^{n^{1-(1/c)^d}}), $$ 
	which is clearly of the form $\tfSPSP{1} \subseteq \tfSPSP{S}$ and of the required size for item (\ref{item:bounded-size}).

	\paragraph{Case 2 (large $p_1$):} If $p_1 > n^{(1/c)^d}$, then notice that the regular formula $\Psi$ can 
	be written in the form 
	$$ \sum_{i=1}^{a_1}\prod_{j=1}^{p_1} f_{i,j} , $$
	where each $f_{i,j}$ is a multilinear $(a_2, p_2, \ldots, a_d, p_d, a_{d+1})$-regular formula. Hence, each $f_{i,j}$ is a polynomial
	of degree bounded by $\prod_{i=2}^d p_i \le n/p_1 < n^{1-(1/c)^d}$. Therefore, expanding each $f_{i,j}$ into 
	a sum of monomials, we obtain a formula $\Phi$ of the form $\tfSPSP{a_1}$ and of size  
	$$ |\Phi| \le a_1 \cdot p_1\cdot \sum_{i=0}^{n^{1-(1/c)^d}} \binom{n}{i} = O(a_1 p_1 n^{n^{1-(1/c)^d}}) = 
	O(S  \cdot n^{n^{1-(1/c)^d}}) . $$  
	This too satisfies item (\ref{item:bounded-size}).

	\paragraph{Case 3 (high degree but small $p_1$):} In this case, we can assume that $\prod_{i=1}^d p_i > n^{1-(1/c)^d}$ and that $p_1 \le n^{(1/c)^d}$. 
	It follows there exists an index 
	$t \in [d-1]$ satisfying
	$$ p_t \le n^{(1/c)^{d+1-t}} \ \text{ and } \ p_{t+1} > n^{(1/c)^{d+1-(t+1)}} = n^{(1/c)^{d-t}},$$
	since otherwise, using $c\geq 3$, we would have that 
	$$ \prod_{i=1}^d p_i \le \prod_{i=1}^d n^{(1/c)^{d+1-i}} = n^{\sum_{i=1}^d (1/c)^{d+1-i}} < n^{\frac{1}{c-1}}  < n^{1-(1/c)^d}, $$
	which contradicts the assumption on the degree.
	
	Notice that we can express $\Psi$ in the form 
	\begin{equation}\label{eq:Psi}
 \ds\sum_{i_1=1}^{a_1}\prod_{j_1=1}^{p_1} \ldots \ds\sum_{i_{t+1}=1}^{a_{t+1}}\prod_{j_{t+1}=1}^{p_{t+1}} 
	f_{i_1,\ldots, i_{t+1},j_1, \ldots , j_{t+1}}, 
	\end{equation}
	where each $f_{i_1,\ldots, i_{t+1},j_1,\ldots,j_{t+1}}$ is a multilinear $(a_{t+2}, p_{t+2}, \ldots, a_d, p_d, a_{d+1})$-regular formula. 
	We shall analyze separately each of the  $f_{i_1,\ldots , i_{t+1}, j_1 , \ldots,  j_{t+1}}$ and the $(a_1, p_1, \ldots, a_{t+1}, p_{t+1}, 1)$-regular formula ``above'' them.
	
	Notice that each $f_{i_1,\ldots, i_{t+1},j_1,\ldots,j_{t+1}}$ is a polynomial
	of degree bounded by $$\prod_{i=t+2}^d p_i < n/p_{t+1} < n^{1-(1/c)^{d-t}}.$$ 
	Therefore, when expressing each $f_{i_1,\ldots, i_{t+1},j_1,\ldots,j_{t+1}}$ as 
	a sum of monomials, its sparsity is upper bounded by
	\begin{equation}\label{eq:sparse-bottom} 
		\ds\sum_{i=0}^{n^{1-(1/c)^{d-t}}} \binom{n}{i} \le  2n^{n^{1-(1/c)^{d-t}}}. 
	\end{equation}
 
	Now, if in \eqref{eq:Psi} we replace each polynomial $f_{i_1,\ldots, i_{t+1},j_1,\ldots,j_{t+1}}$ with a new variable $y_{i_1,\ldots, i_{t+1},j_1,\ldots,j_{t+1}}$, then we get an  $(a_1, p_1, \ldots, a_{t+1}, p_{t+1}, 1)$-regular formula in the $y$ variables $$\Phi_1 = \ds\sum_{i_1=1}^{a_1}\prod_{j_1=1}^{p_1} \ldots \ds\sum_{i_{t+1}=1}^{a_{t+1}}\prod_{j_{t+1}=1}^{p_{t+1}} 	y_{i_1,\ldots, i_{t+1},j_1,\ldots,j_{t+1}}.$$
	It is clear that $\Psi$ is the composition
	of $\Phi_1$ with the assignment $y_{i_1,\ldots, i_{t+1},j_1,\ldots,j_{t+1}}=f_{i_1,\ldots, i_{t+1},j_1,\ldots,j_{t+1}}$.
	
	By applying the Depth Reduction (Lemma~\ref{lem:squeeze}) repeatedly to $\Phi_1$, we obtain that
	$\Phi_1$ becomes a multilinear $\left( \prod_{i=1}^{t+1} a_i^{\pi_{i-1}}, \pi_{t+1}, 1 \right)$-regular
	formula $\Phi_2$, where $\pi_k = \prod_{i=1}^k p_i$, for any $1 \le k \le d$ (and $\pi_0=1$).
	
	Composing $\Phi_2$ with the assignment $y_{i_1,\ldots, i_{t+1},j_1,\ldots,j_{t+1}}=f_{i_1,\ldots, i_{t+1},j_1,\ldots,j_{t+1}}$, we obtain the following 
	depth-$4$ regular formula $\Phi$ (after some proper relabelling): 
	\begin{equation}\label{eq:red-form}
		\Phi = \ds\sum_{i=1}^{M}\prod_{j=1}^{\pi_{t+1}} f_{i,j}, \ \text{ where } \ M = \ds\prod_{i=1}^{t+1} a_i^{\pi_{i-1}}.
	\end{equation}
	Notice that, by our choice of the parameter $t$, we have 
	$$ M = \ds\prod_{i=1}^{t+1} a_i^{\pi_{i-1}} \le S^{\pi_t} = S^{\prod_{i=1}^t p_i} \le S^{\prod_{i=1}^t n^{(1/c)^{d+1-i}}} < S^{n^\alpha},
	\ \text{ where } \ \alpha = \left(\frac{1}{c}\right)^{d-t} \frac{1}{c-1} .  $$
	Since, by equation~\eqref{eq:sparse-bottom}, each $f_{i,j}$ has sparsity bounded by 
	$2n^{n^{1-(1/c)^{d-t}}} = 2n^{n^{1-(c-1)\alpha}}$,
	we have that $\Phi$ is a $\tfSPSP{M}$ formula of size bounded by:
	$$ |\Phi| \le M \cdot \pi_{t+1} \cdot 2n^{1-(c-1)\alpha} \le M \cdot n \cdot 2n^{1-(c-1)\alpha} .   $$
	This satisfies the conditions of item (\ref{item:bounded-fan-in}), and so this concludes the proof of the theorem.
\end{proof}

\subsection{PIT for Regular Formulas}

In this section, we construct our hitting set for regular formulas. Since the reduction done in Theorem~\ref{thm:squeeze} 
reduces a multilinear depth-$d$ regular formula to one of two types of depth-$4$ formulas,
our hitting set will be the union of two hitting sets, $\cF_d$ and $\cG_d$, each one designed to hit a specific type.

\begin{thm}[Theorem~\ref{thm:intro:hitting-set-regular}, restated]\label{thm:hitting-set-regular}
	For $d\geq 2$, let $\cC_d$ be the class of multilinear polynomials computed by 
	$(a_1, p_1, a_2, p_2, \ldots, a_d, p_d, a_{d+1})$-regular formulas of size $S \le 2^{n^\delta}$ 
	computing a multilinear polynomial $f(x_1, \ldots, x_n)$, where $\delta = \frac{1}{5^{d+1}}$. 
	Then, there exists a hitting set $\cH_d$ of size 
	$|\cH_d| = 2^{\tilde{O}(n^{1-\delta/3})}$ for $\cC_d$, that can be constructed in time $\poly(|\cH_d|)$.
\end{thm}

The proof follows by combining Theorem~\ref{thm:squeeze} with the hitting set guaranteed by Theorem~\ref{thm:hitting-set-depth-$4$-general}.

\begin{proof}
	Let $f(x_1, \ldots, x_n)$ be a polynomial computed by a formula $\Psi \in \cC_d$.
	By Theorem~\ref{thm:squeeze}, with the constant $c =5$, there are two cases of the depth reduction to analyze. For each case we will give a hitting set (using Theorem~\ref{thm:hitting-set-depth-$4$-general}) and thus the union of the sets will be a hitting set for $\cC_d$.
	
		\paragraph*{Case 1:} For $M=S\leq 2^{n^\delta}$, there exists a $\tfSPSP{M}$ formula $\Phi$ of size $|\Phi| = O( S \cdot n^{n^{1-(1/5)^d}}) = 
	O( 2^{n^\delta} \cdot n^{n^{1-5\delta}})$ computing $f(x_1, \ldots, x_n)$.
	By Theorem~\ref{thm:hitting-set-depth-$4$-general}, there exists a hitting set $\cH'$ that hits all such non-zero formulas with 
	$$|\cH'| =  2^{\tilde{O}( n^{2/3} \cdot \log M \cdot (\log |\Phi|)^{1/3} )}.$$
	Observe that 
	$$(\log M)^3 \cdot \log|\Phi| = n^{3\delta}\cdot (n^\delta + n^{1-5\delta} \cdot \log n)  = n^{4\delta} + n^{1-2\delta}\log n = \tilde{O}(n^{1-2\delta}).$$
	Hence $$|\cH'| =  2^{\tilde{O}( n^{2/3} \cdot \log M \cdot (\log |\Phi|)^{1/3} )} = 2^{\tilde{O}(n^{2/3}\cdot n^{(1-2\delta)/3})}=2^{\tilde{O}(n^{1 - 2\delta/3})}.$$

	\paragraph*{Case 2:} There exists $t \in [d-1]$ such that for $ \alpha_t = \frac{1}{4} \cdot \left(\frac{1}{5} \right)^{d-t} \leq \frac{1}{20}$,   
	there exists a multilinear $\tfSPSP{M}$ formula $\Phi$ computing $f(x_1, \ldots, x_n)$, 
	where the top fan-in $ M = S^{n^{\alpha_t}} = 2^{n^{\delta + \alpha_t}}$ and the size is bounded by 
	$|\Phi| \le 2Mn \cdot n^{n^{1-4\alpha_t}}. $ Again, by Theorem~\ref{thm:hitting-set-depth-$4$-general}, there exists a hitting set $\cH''$ that hits all such non-zero formulas with 
	$$|\cH''| =  2^{\tilde{O}( n^{2/3} \cdot \log M \cdot (\log |\Phi|)^{1/3} )}.$$
	We now have that 
	$$(\log M)^3 \cdot \log|\Phi| = \tilde{O}\left( n^{3(\delta + \alpha_t)} \cdot (n^{\delta + \alpha_t} + n^{1-4\alpha_t} )\right) =  \tilde{O}\left( n^{4(\delta + \alpha_t)} + n^{1+3\delta-\alpha_t} \right) =   \tilde{O}\left( n^{1-\delta}\right),$$
	where the last equality holds as, by our choice of parameters, for all $t\in [d-1]$, $5\delta + 4\alpha_t<1$ and $4\delta < \alpha_t$. This implies that
	\[ |\cH''| =  2^{\tilde{O}( n^{2/3} \cdot \log M \cdot (\log |\Phi|)^{1/3} )} =  2^{\tilde{O}(n^{2/3} \cdot n^{(1-\delta)/3})} = 2^{\tilde{O}(n^{1-\delta/3})} . \qedhere \]
\end{proof}

We note that we did not attempt to optimize the parameters in the theorem as, using our current proof, the exponent is going to be of the form $n^{1-1/\exp(d)}$ anyway.

\section{Lower Bounds for Bounded Depth Multilinear Formulas}
\label{sec:lower-bounds}

As we noted earlier, the connection between construction of hitting sets and lower bounds for explicit polynomials is well established. 
The following theorem was proved by Heintz and Schnorr \cite{HeintzSchnorr80} and Agrawal \cite{Agrawal05}, albeit we cite only 
a special case which matches our use of it:

\begin{thm}[A special case of \cite{HeintzSchnorr80,Agrawal05}]
\label{thm:lowerbound-from-hitting}
Suppose there is a black-box deterministic PIT algorithm for a class $\cC$ of multilinear circuits, that outputs a hitting set $\cH$ of size $|\cH| = 2^{n^\alpha}<2^{n}$ and runs in time $\poly(|\cH|)$, such that $\cH$ hits circuits of size at most $2^{n^\delta}$. Then, there exists a multilinear polynomial $f(x_1,\ldots,x_n)$ such that any circuit from the class $\cC$ computing $f$ must be of size at least $2^{n^{\delta}}$, and the coefficients of $f$ can be found in time $2^{O(n)}$.
\end{thm}

Theorem~\ref{thm:lowerbound-from-hitting} is proved by finding a non-zero polynomial $f(x_1,\ldots,x_n)$ which vanishes on the entire 
hitting set $\cH$ of size $2^{n^\alpha}$, and then, by definition, $f$ cannot have circuits of size 
$2^{n^\delta}$. Finding $f$ amounts to finding a non-zero solution to a homogenous system of 
linear equations whose unknowns are the coefficients of the $2^{n}$ possible multilinear monomials in 
$x_1,\ldots,x_n$. As long as $2^n > |\cH| = 2^{n^\alpha}$, a non-zero solution is guaranteed to exist.

Our lower bounds now follow as a direct application of our hitting set constructions and Theorem~\ref{thm:lowerbound-from-hitting}.

\begin{proof}[Proofs of Corollaries \ref{cor:intro:lowerbound-depth-$3$}, \ref{cor:intro:lowerbound-depth-$4$} and \ref{cor:intro:lowerbound-regular}]
In light of Theorem~\ref{thm:lowerbound-from-hitting}, we only need to pick $\delta$ so that the hitting sets we constructed have size less than $2^n$. The appropriate choices, by Theorems \ref{thm:intro:hitting-set-depth-$3$}, \ref{thm:intro:hitting-set-depth-$4$} and \ref{thm:intro:hitting-set-regular}, respectively, can be seen to be $\delta=1/2 - O(\log \log n / \log n)$ (for depth-$3$), $\delta=1/4-O(\log \log n / \log n)$ (for depth-$4$) and $\delta =  \frac{1}{5^{\lfloor d/2 \rfloor+1}}  = O\left(\frac{1}{\sqrt{5}^d}\right)$
(for depth-$d$ regular formulas).
\end{proof}


\section{Conclusion and Open Questions}\label{sec:conclusion}
\label{sec:open}

We conclude this paper with some obvious open problems. First, as noted in Section~\ref{sec:results}, the lower bounds that we get from our hitting sets are not as good as the best lower bounds for these models. Can one improve our construction to yield lower bounds matching the best known lower bounds?

Currently, the size of the hitting set that we get for depth-$d$ regular multilinear formulas is roughly $\exp(n^{1-1/\exp(d)})$. Can the bound be improved to $\exp(n^{1-\Omega(1/d)})$ ? In our proof the reason for this exponential loss is that we reduce the regular formula to a $\tfSPSP{M}$ formula of size $S$ and we need $M$ and $S$ to satisfy (because of Theorem~\ref{thm:hitting-set-depth-$4$-general}) $(\log M)^3\cdot \log S = o(n)$. In particular, if $M=2^{n^{\delta}}$ and $S=2^{n^{\gamma}}$ then we require that $3\delta + \gamma <1$. Notice that, in the depth reduction theorem (Theorem~\ref{thm:squeeze}), if we start with a regular formula of size $2^{n^\delta}$ then, if we break the formula at layer $t$, we roughly get a top fan-in of $M=2^{n^\delta \cdot p_1\cdot p_2 \cdots p_t}$ and bottom sparsity of (roughly) $\exp(n^{1-p_{t+1}})$. This gives a size upper bound of (roughly) $S=2^{n^\delta \cdot p_1\cdot p_2 \cdots p_t} \cdot \exp(n^{1-p_{t+1}})$. To match the requirement $(\log M)^3\cdot \log S = o(n)$, we get that $p_{t+1}$ must be larger than $3p_1\cdots p_t$. This leads to an argument in which we require the degree of the product gates to increase exponentially. This is more or less the cause of the exponential loss in our argument. 

Finally, another natural question is to extend our argument from depth-$d$ regular multilinear formulas to arbitrary depth-$d$ multilinear formula. 

\section*{Acknowledgments}

The authors would like to thank Zeev Dvir and Avi Wigderson for helpful discussions during the
course of this work.


\bibliographystyle{alpha}

\bibliography{bibliography}
\newpage

\appendix

\end{document}